\documentclass[11pt,a4paper]{amsart}
\usepackage{fullpage}
\usepackage{microtype}
\usepackage[tt=false]{libertine} 
\usepackage[libertine]{newtxmath}
\usepackage{amsmath}
\usepackage{amsthm}
\usepackage{mathtools} 
\usepackage{algorithm}
\usepackage{algpseudocode}
\usepackage{enumitem}

\usepackage{tikz}
\usetikzlibrary{arrows,arrows.meta,backgrounds,calc,fit,decorations.pathreplacing,decorations.markings,shapes.geometric}
\def\borderColor{blue!60}
\def\scale{0.6}
\def\nodeDist{1.4cm}
\tikzstyle{internal} = [draw, fill, shape=circle]
\tikzstyle{external} = [shape=circle]
\tikzstyle{square}   = [draw, fill, rectangle]
\tikzstyle{triangle} = [draw, fill, regular polygon, regular polygon sides=3, inner sep=3pt]
\tikzstyle{pentagon} = [draw, fill, regular polygon, regular polygon sides=5, inner sep=2pt, minimum size=14pt]

\usepackage[margin=1cm]{caption} 
\usepackage{subfig}

\usepackage[thinlines]{easytable}

\tikzset{every fit/.append style=text badly centered}

\usepackage[bookmarks=true,hypertexnames=false,pagebackref]{hyperref}
\hypersetup{colorlinks=true, citecolor=blue, linkcolor=red, urlcolor=blue}

\usepackage{ifthen}

\usepackage{cleveref}

\usepackage[textsize=tiny]{todonotes}
\usepackage{mleftright}
\usetikzlibrary{positioning,chains,fit,shapes,calc}
\usetikzlibrary{trees}
\usetikzlibrary{decorations.pathreplacing}
\usetikzlibrary{decorations.pathmorphing}
\usetikzlibrary{decorations.markings}
\tikzset{>=latex} 

\usepackage{cool}
\Style{DSymb={\mathrm d},DShorten=true,IntegrateDifferentialDSymb=\mathrm{d}}

\newcommand{\tp}[1]{{\left( #1 \right)}}
\newcommand{\sqtp}[1]{{\left[ #1 \right]}}
\newcommand{\cmid}{\,:\,}

\def\*#1{\mathbf{#1}}
\def\+#1{\mathcal{#1}}
\def\-#1{\mathrm{#1}}
\def\=#1{\mathbb{#1}}

\newcommand{\abs}[1]{\left\vert#1\right\vert}

\newcommand{\floor}[1]{\lfloor#1\rfloor}
\newcommand{\set}[1]{\left\{#1\right\}}
\newcommand{\eps}{\varepsilon}

\newcommand{\arity}{\operatorname{arity}}

\newcommand{\trans}[4]{\ensuremath{\left[\begin{smallmatrix} #1 & #2 \\ #3 & #4 \end{smallmatrix}\right]}}
\newcommand{\transpose}[1]{#1^\texttt{T}}

\newcommand{\Holant}{\operatorname{Holant}}
\newcommand{\holant}[2]{\ensuremath{\Holant\left(#1\mid #2\right)}}

\newcommand{\ol}{\overline}

\newtheorem{theorem}{Theorem}

\newtheorem{lemma}[theorem]{Lemma}

\newtheorem{proposition}[theorem]{Proposition}

\newtheorem{definition}[theorem]{Definition}
\newtheorem*{remark}{Remark}

\crefname{theorem}{Theorem}{Theorems}
\crefname{observation}{Observation}{Observations}
\crefname{claim}{Claim}{Claims}
\crefname{condition}{Condition}{Conditions}
\crefname{algorithm}{Algorithm}{Algorithms}
\crefname{property}{Property}{Properties}
\crefname{example}{Example}{Examples}
\crefname{fact}{Fact}{Facts}
\crefname{lemma}{Lemma}{Lemmas}
\crefname{corollary}{Corollary}{Corollaries}
\crefname{definition}{Definition}{Definitions}
\crefname{remark}{Remark}{Remarks}
\crefname{proposition}{Proposition}{Propositions}
\crefname{equation}{equation}{equations}
\crefname{enumi}{Case}{Case}
\creflabelformat{enumi}{(#2#1#3)}

\makeatletter
\def\prob#1#2#3{\goodbreak\begin{list}{}{\labelwidth\z@ \itemindent-\leftmargin
      \itemsep\z@  \topsep6\p@\@plus6\p@
      \let\makelabel\descriptionlabel}
  \item[\textbf{Name}]#1
  \item[\textbf{Instance}]#2
  \item[\textbf{Output}]#3
  \end{list}}
\makeatother

\newcommand{\defeq}{:=}

\newcommand{\ind}[2]{\#\mathtt{Ind}\tp{#1,#2}}

\title{Zeros of Holant problems: locations and algorithms}

\author{Heng Guo}
\address[Heng Guo]{School of Informatics, University of Edinburgh, Informatics Forum, Edinburgh, EH8 9AB, United Kingdom.}
\email{hguo@inf.ed.ac.uk}

\author{Chao Liao}
\address[Chao Liao]{Department of Computer Science and Engineering, Shanghai
Jiao Tong University, No.800 Dongchuan Road, Minhang District, Shanghai, China.}
\email{chao.liao.95@gmail.com}

\author{Pinyan Lu}
\address[Pinyan Lu]{ITCS, Shanghai University of Finance and Economics, No.100
Wudong Road, Yangpu District, Shanghai, China.}
\email{lu.pinyan@mail.shufe.edu.cn}

\author{Chihao Zhang}
\address[Chihao Zhang]{ITCSC, The Chinese University of Hong Kong,
  Sha Tin, N.T., Hong Kong, China.}
\email{chihao.zhang@gmail.com}

\begin{document}

\begin{abstract}
  We present fully polynomial-time (deterministic or randomised) approximation schemes for Holant problems,
  defined by a non-negative constraint function satisfying a generalised second order recurrence modulo a couple of exceptional cases.
  As a consequence, any non-negative Holant problem on cubic graphs has an efficient approximation algorithm unless the problem is equivalent to approximately counting perfect matchings,
  a central open problem in the area.
  This is in sharp contrast to the computational phase transition shown by 2-state spin systems on cubic graphs.
  Our main technique is the recently established connection between zeros of graph polynomials and approximate counting.
  We also use the ``winding'' technique to deduce the second result on cubic graphs.
\end{abstract}

\maketitle

\section{Introduction}

Great progress has been made recently in the classification of counting problems.
One major achievement is the full dichotomy for counting constraint satisfaction problems (CSPs) \cite{Bul13,DR13}, even with complex weights \cite{CC17}.
However, such a classification is for exact counting, and for approximation,
even to move beyond some rather modest model seems quite difficult.

Holant problems \cite{CLX11} are a framework of expressing counting problems motivated by Valiant's holographic algorithms \cite{Val08}.
The ``Holant'' is a partition function on graphs where edges are variables and vertices are constraint functions.
The benefit of this choice is the ability to express problems like perfect matchings,
which are provably not expressible in certain CSP-like vertex models \cite{FLS07,DGLRS12,Sch13}.
We parameterise Holant problems by the set of constraint functions that can be put on vertices.
Similar to the success of classifying counting CSPs,
exact classifications have been obtained for Holant problems defined by any set of complex-weighted symmetric Boolean functions \cite{CGW16},
and various progresses have been made to go beyond \cite{CLX18,LW18,Bac18}.

In this paper, we make progress towards understanding the complexity of approximating symmetric Boolean Holant problems with non-negative weights.
Let $G=(V,E)$ be a graph,
$\pi:V\rightarrow \+F$ be an assignment from the set of vertices $V$ to a set of functions $\+F$,
and $f_v=\pi(v)$ is the constraint function $\{0,1\}^{\deg(v)}\rightarrow\=C$ associated with the vertex $v$.
The ``Holant'' is defined as follows:
\begin{align} \label{eqn:Holant-def}
  Z(G;\pi)\defeq \sum_{\sigma\in\set{0,1}^E} \prod_{v\in V} f_v (\sigma|_{E(v)}),
\end{align}
where ${E(v)}$ is the set of adjacent edges of $v$,
and $\sigma|_{E(v)}$ is the restriction of $\sigma$ on $E(v)$.
We use the shorthand $Z(G)$ or $Z$ when $G$ and $\pi$ are clear from the context.

%

We call a Boolean constraint function $f$ \emph{symmetric},
if $f(\*x)$ depends only on the hamming weight $\abs{\*x}$ and is invariant under permutations of the indices.
For a symmetric $f$ of arity $d$,
we associate it with a \emph{signature} $[f_0,f_1,\dots,f_d]$,
where $f_i=f(\*x)$ if $\abs{\*x}=i$.
We may use the term ``constraint function'' and ``signature'' interchangeably.
For example, if $f$ is the ``exact-one'' function, namely $f=[0,1,0,\dots,0]$,
then $Z(G)$ counts the number of perfect matchings in $G$;
and if $f$ is the Boolean OR function, namely $f=[0,1,1,\dots,1]$,
then $Z(G)$ counts the number of edge covers in $G$.

We focus on a fairly expressive family of symmetric functions satisfying generalised second-order recurrences.
More precisely, we say $f=[f_0,f_1,\dots,f_{d}]$ satisfies a generalised second-order recurrence,
if there exist real constants $(a,b,c)\neq (0,0,0)$ such that $a f_k + b f_{k+1} + c f_{k+2} = 0 $ for all $0\le k\le d-2$.
Denote by $\Holant(f)$ the computational problem of evaluating $Z(G)$ where every vertex is associated with the signature $f$.
In particular, the input to $\Holant(f)$ must be $d$-regular, where $d$ is the arity of $f$.
Our main theorem is the following.

\begin{theorem}\label{thm:main}
  Let $f=[f_0,f_1,\dots,f_{d}]$ be a symmetric constraint function of arity $d\ge 3$ satisfying generalised second-order recurrences,
  and $f_i\ge 0$ for all $0\le i\le d$.
  There is a fully polynomial-time (deterministic or randomised) approximation algorithm for $\Holant(f)$,
  unless, up to a non-zero factor, $f$ or its reversal is in one of the following form:
  \begin{itemize}
    \item $[0,\lambda\sin\frac{\pi}{d},\lambda^{2}\sin\frac{2\pi}{d},\dots,\lambda^{i}\sin\frac{i\pi}{d},\dots,0]$ for some $\lambda>0$;
    \item $[0,1,0,\lambda,0,\dots,0,\lambda^{\frac{d-2}{2}},0]$ if $d$ is even, or $[0,1,0,\lambda,0,\dots,0,\lambda^{\frac{d-1}{2}}]$ if $d$ is odd, for some $0\le\lambda < 1$.
  \end{itemize}
  Moreover, in the latter case, approximating $\Holant(f)$ is equivalent to approximately counting perfect matchings in general graphs.
\end{theorem}

Understanding the complexity of signatures with second-order recurrences is the cornerstone in the exact counting classifications.
Since satisfying first-order recurrences implies that the function is degenerate,
these constraint functions are the first class satisfying a recurrence relation with non-trivial complexity.
More concretely, this family includes many interesting special cases:
\begin{itemize}
  \item Matchings and perfect matchings. The functions are $[1,1,0,0,\dots,0]$ and $[0,1,0,0,\dots,0]$, respectively, with $(a,b,c)=(0,0,1)$.
  \item Even subgraphs, whose functions are $[1,0,1,0,\dots]$ with $(a,b,c)=(1,0,-1)$.
    More generally, we may put weights on even and odd degree vertices, and the functions become $[x,y,x,y,\dots]$ for some $x,y\ge 0$.
  \item Edge covers, whose functions are $[0,1,1,\dots,1]$ with $(a,b,c)=(0,1,-1)$.
  \item Fibonacci gates, namely $f$ of arity $d$ such that $f_{i+2}= b f_{i+1}+f_i$ for all $i\le d-2$.
  \item All ternary symmetric functions.
\end{itemize}
For approximate counting, polynomial-time approximation algorithms are known only for a few special cases,
such as counting matchings \cite{JS89}, weighted even subgraphs \cite{JS93}, counting edge covers \cite{LLL14}, and a weighted version of Fibonacci gates \cite{LWZ14}.
However, neither the Markov chain Monte Carlo approach \cite{JS89,JS93} (including its ``winding'' extension \cite{McQ13,HLZ16}),
nor the correlation decay approach \cite{LWZ14,LLL14},
appears to be powerful enough to handle all functions in this family.
On the other hand, \Cref{thm:main} covers almost all problems in this family,
and most of the exceptional cases are equivalent to counting perfect matchings, a central open problem in approximate counting (see, for example, \cite{DJM17, SVW18} on partial progresses and barriers).
Efficient approximate counting for perfect matchings is only known in the bipartite case \cite{JSV04}.

As a consequence, we have an algorithm for all non-negative Boolean Holant on cubic graphs, unless the problem is equivalent to counting perfect matchings.

\begin{theorem}  \label{thm:main-cubic}
  Let $f=[f_0,f_1,f_2,f_3]$ be a symmetric constraint function of arity $3$ where $f_i\ge0$ for all $0\le i\le 3$ .
  $\Holant(f)$ has a fully polynomial-time (deterministic or randomised) approximation algorithm,
  unless $f$ or its reversal, up to a non-zero factor, is $[0,1,0,\lambda]$ for some $0\le\lambda < 1$.
  In the exceptional case, approximating $\Holant(f)$ is equivalent to approximately counting perfect matchings in general graphs.
\end{theorem}


We remark that \Cref{thm:main-cubic} is in sharp contrast to the computational phase transition phenomenon,
as demonstrated by 2-state spin systems on cubic graphs \cite{GJP03,SS14,GSV16,LLY13,SST14}, even without external fields.
For spin systems, a clear and sharp threshold between approximable and hard to approximate is established,
whereas for Holant problems on cubic graphs, there seems to be no such transition.

\subsection{Our techniques}

Our algorithm combines a number of ingredients:
\begin{itemize}
  \item Barvinok's approach to approximate partition functions via Taylor expansions \cite{Bar16}.
    This approach was sharpened by Patel and Regts \cite{PR17a} to run within polynomial-time.
  \item In order to apply Barvinok's approach, one has to have some rather precise knowledge of the zeros of the corresponding graph polynomials.
    For Holant problems, Ruelle \cite{Rue71,Rue99a,Rue99b} has developed a systematic approach of bounding the zeros of the partition function
    via analysing polynomials associated locally with vertices, under the disguise of ``graph-counting polynomials''.
  \item On top of combining Ruelle's and Barvinok's approaches,
    we also employ holographic transformations a la Valiant \cite{Val08},
    which is necessary to cover all cases in \Cref{thm:main}.
\end{itemize}

Although none of these ingredients is new, the main contribution of our work is to combine them together (with reworks if necessary),
and a thorough analysis of the zeros of functions with generalised second-order recurrence.
To be more specific, for a symmetric signature $f=[f_0,\dots,f_d]$ of arity $d$,
define the ``local'' polynomial of $f$ as
\begin{align}\label{eqn:local-f-poly}
  P_f(z)\defeq\sum_{i=0}^d\binom{d}{i}f_i\cdot z^i.
\end{align}
We may also view $P_f(z)$ as the polynomial for a single vertex with $d$ dangling edges.
For some $\eps>0$, we call a polynomial $P(z)$ \emph{$H_{\eps}$-stable}, if $P(z)\neq 0$ as long as $\Re z\ge -\eps$.
Then one of our main technical tool (see \Cref{thm:roots-Holant}) says that if $P_f(z)$ is $H_{\eps}$-stable for some $\eps>0$,
then a polynomial-time approximation algorithm exists for $\Holant(f)$.

In general, to apply Barvinok's method to approximate counting,
one needs to deal with the zeros of the whole partition function,
which is usually not an easy task.
Previous applications appeal to some powerful tools such as the Lee-Yang theorem from statistical physics \cite{LSS17},
or the resolution of a long-standing conjecture \cite{PR17c}.
In contrast, our approach requires only analysing some low degree polynomials and is much easier to apply.

To go from \Cref{thm:main} to \Cref{thm:main-cubic},
we also need to deal with cases not covered by \Cref{thm:main}, which cannot be solved using zeros of Holant problems.
These exceptional cases are handled by the ``winding'' technique \cite{McQ13,HLZ16} with Markov chains.

\section{Ruelle's method on zeros of Holant problems}

Ruelle \cite{Rue71,Rue99a, Rue99b} (building upon the ``Asano contraction'' \cite{Asa70})
has developed a systematic approach to bound zeros of the so-called ``graph-counting polynomials''.
As we will see later, these polynomials coincide with unweighted Holant problems.

With a little abuse of notation, let $Z(G;f)$ be the partition function defined by \eqref{eqn:Holant-def} where $f_v=f$ for all $v\in V$,
and stratify $Z(G;f)$ by the number of edges chosen as follows:
\begin{align}\label{eqn:Holant-k-def}
  Z_k(G;f)\defeq\sum_{\sigma\in\set{0,1}^E\text{ and }\abs{\sigma}=k}\;\prod_{v\in V} f(\sigma|_{E(v)}).
\end{align}
Define $Z_k(G;\pi)$ similarly, and again, $G$ and $f$ may be omitted when they are clear from the context.

Let $\abs{E}=m$.
Then $Z=Z(G;f)$ can be rewritten as the evaluation of the polynomial
\begin{align}  \label{eqn:Holant-polynomial}
  P_G(z)\defeq\sum_{i=0}^m Z_i \cdot z^i
\end{align}
at $z=1$.
Namely $Z=P_G(1)$.
When $f$ is a symmetric $0/1$ function,
then \eqref{eqn:Holant-polynomial} is the same as the ``graph-counting'' polynomial defined by Ruelle \cite{Rue99b}.



Ruelle's method has two main ingredients.
Firstly we want to relate zeros of a univariate polynomial with those of its polar form.
For a polynomial $P(z)=\sum_{i=0}^{d'} a_i z^i$ of degree $d'\le d$,
its $d$th polar form with variables $\*z=(z_1,\dots,z_d)$ is
\begin{align*}
  \widehat{P}(\*z)\defeq\sum_{I\subseteq[d]}\frac{a_{\abs{I}}}{\binom{d}{\abs{I}}}z_I,
\end{align*}
where $a_i =0$ if $i> d'$, $[d]$ denotes $\{1,2,\dots,d\}$,
and for an index set $I$, $z_I=\prod_{i\in I}z_i$.
For example, the polar form of $P_f(z)$ (recall \eqref{eqn:local-f-poly}) is,
\begin{align*}
  \widehat{P}_f(\*z)\defeq\sum_{I\subseteq[d]}f_{\abs{I}}z_I.
\end{align*}
The polar form $\widehat{P}(\*z)$ is the unique multi-linear symmetric polynomial of degree at most $d'$ such that
$\widehat{P}(z,z,\dots,z) = P(z)$.
When $d'<d$, we view $P(z)$ as a degenerate case,
and it has zeros at $\infty$ with multiplicity $d-d'$.

Let $H$ be a region in $\=C$.
We say a polynomial $P(\*z)$ in $d\ge 1$ variables is \emph{$H$-stable} if $P(\*z)\neq 0$ whenever $z_1,\dots,z_d\in H$.
We will be particularly interested in $H_{\eps}$-stableness where $H_{\eps}$ is the half-plane:
\begin{align*}
  H_{\eps}=\set{z\in\=C\mid\Re z \ge -\eps},
\end{align*}
and $\eps>0$.
The Grace-Szeg\H{o}-Walsh coincidence theorem \cite{Gra02,Sze22,Wal22} has the following immediate consequence.
\begin{proposition}  \label{prop:Grace-Walsh-Szego}
  A univariate polynomial $P(z)$ is $H_{\eps}$-stable if and only if its polar form $\widehat{P}(\*z)$ is $H_{\eps}$-stable.
\end{proposition}
\Cref{prop:Grace-Walsh-Szego} actually applies to an arbitrary circular domain in $\=C$,
but we will only need it for~$H_{\eps}$.

The next ingredient is the Asano contraction \cite{Asa70}, as extended by Ruelle \cite{Rue71}.

\begin{proposition}  \label{prop:Asano}
  Let $K_1$ and $K_2$ be closed subsets of the complex plane $\=C$, which do not contain $0$.
  If the complex polynomial
  \begin{align*}
    \alpha+\beta z_1+\gamma z_2+\delta z_1z_2
  \end{align*}
  does not vanish for any $z_1\not\in K_1$ and $z_2\not\in K_2$,
  then
  \begin{align*}
    \alpha+\delta z
  \end{align*}
  does not vanish for any $z\not\in -K_1\cdot K_2$.
\end{proposition}
We refer interested readers to \cite{Rue71} for a very elegant proof of \cref{prop:Asano}.

Let the $\delta$-strip of $[0,1]$ be
\begin{align*}
  \set{z\in\=C\mid \abs{\Im z}\le\delta\text{ and }-\delta\le\Re z\le 1+\delta}.
\end{align*}
\begin{lemma}  \label{lem:HH}
  For any $\eps>0$, the complement of $-H_{\eps}\cdot H_{\eps}$ contains a $\delta$-strip of $[0,1]$ for some $\delta>0$ depending only on $\eps$.
\end{lemma}
\begin{proof}
  An equivalent way to write $H_{\eps}$ is
  \begin{align*}
    H_{\eps}=\set{\rho e^{i\theta}\mid\rho\ge-\frac{\eps}{\cos \theta}\text{ for $\theta\in\tp{\frac{\pi}{2},\frac{3\pi}{2}}$}}.
  \end{align*}
  Thus,
  \begin{align*}
    -H_{\eps}\cdot H_{\eps} & = \set{\rho_1\rho_2 e^{i(\theta_1+\theta_2+\pi)}\mid
    \rho_i\ge-\frac{\eps}{\cos \theta_i}\text{ for $\theta_i\in\tp{\frac{\pi}{2},\frac{3\pi}{2}}$ and $i\in\{1,2\}$}} \\
    & = \set{\rho e^{i(\theta_1+\theta_2+\pi)}\mid
    \rho\ge \frac{\eps^2}{\cos \theta_1\cos\theta_2}\text{ for $\theta_1,\theta_2\in\tp{\frac{\pi}{2},\frac{3\pi}{2}}$}}\\
    & = \set{\rho e^{i\theta}\mid
    \rho\ge \frac{\eps^2}{\tp{\cos \frac{\theta-\pi}{2}}^2}\text{ for $\theta\in\tp{0,2\pi}$}}\\
    & = \set{\rho e^{i\theta}\mid
    \rho\ge \frac{2\eps^2}{1-\cos \theta}\text{ for $\theta\in\tp{0,2\pi}$}},
  \end{align*}
  where the third line is because $\cos\theta_1\cos\theta_2$ is maximised at $\theta_1=\theta_2$ if their sum is fixed.
  It is easy to check that $\delta=\eps^2/2$ suffices for the claim.
\end{proof}

Now we are ready to state a very useful lemma.
\begin{lemma}  \label{thm:Ruelle}
  Let $f$ be a symmetric signature of arity $\Delta$.
  If the local polynomial $P_f(z)$ is $H_{\eps}$-stable for some $\eps>0$,
  then the global polynomial $P_G(z)$ has no zero in the $\delta$-strip of $[0,1]$,
  where $\delta$ is a constant depending only on $\eps$.
\end{lemma}
\begin{proof}
  We construct $G=(V,E)$ as follows.
  Start with a collection of vertices $v\in V$, each with $\Delta$ dangling half-edges $(e_i^v)_{i\in[\Delta]}$.
  Call this graph $G_0$, and connect dangling half-edges $e_i^v$ and $e_j^u$ sequentially for each edge $(u,v)\in E$.
  This gives a sequence of graphs $G_1,\dots,G_{\abs{E}}=G$.
  The polynomial of $G_0$ is $P_{G_0}(z)=\prod_{v\in V} P_{v}(z)$, where $P_v=p_f$,
  and consider the multivariate version $\widehat{P}_{G_0}(\*z)=\prod_{v\in V}\widehat{P}_v(\*z^v)$,
  where $\widehat{P}_v = \widehat{P}_f$ and $\*z^v$ is the local variables corresponding to $v$.
  Since $P_f(z)$ is $H_{\eps}$-stable, by \Cref{prop:Grace-Walsh-Szego}, $\widehat{P}_f(\*z)$ is as well,
  and so is $\widehat{P}_{G_0}(\*z)$.
  Suppose from $G_i$ to $G_{i+1}$, $e^v_i$ is connected with $e^u_j$.
  Then the transformation from $\widehat{P}_{G_i}$ to $\widehat{P}_{G_{i+1}}$ is
  exactly the Asano contraction as in \Cref{prop:Asano} applied to $z^v_i$ and $z^u_j$.
  Let $K$ be the complement of $H_{\eps}$.
  At the end of this procedure we obtain $G$ and the polynomial $\widehat{P}_G(\*z)$ does not vanish on the complement of $-K\cdot K$.
  It implies that the same is true for the univariate $P_G(z)$.
  By \Cref{lem:HH}, the complement of $-K\cdot K$ contains a $\delta$-strip of $[0,1]$,
  and this $\delta$ depends only on~$\eps$.
\end{proof}

We note that it is necessary to have some slack $\eps$ in \Cref{thm:Ruelle}.
One example is counting even subgraphs, namely the constraint $f$ is $[1,0,1,0,\dots]$.
Although all zeros of $P_f$ lie on the imaginary axis,
the zeros of $P_G$ can in fact be dense on the unit circle.
To see this, let $G$ be a cycle of length $n$.
Then $P_G(z)=1+z^n$ as there are only two even subgraphs.
The zeros thereof are dense on the unit circle as $n$ varies.

\Cref{thm:Ruelle} can be easily generalised to a set of functions,
if there is an $\eps>0$ such that all of the local polynomials are $H_{\eps}$-stable.
A univariate polynomial is called \emph{Hurwitz stable} if all of its zeros are in the open left half-plane.
For a fixed $f$, clearly if $P_f(z)$ is Hurwitz stable, then there is some $\eps>0$ such that $P_f(z)$ is $H_{\eps}$-stable.
However, Hurwitz stability is not enough to derive the same conclusion of \Cref{thm:Ruelle} for an infinite set of functions.

%

\section{Barvinok's algorithm}

Our interest in Ruelle's method, \Cref{thm:Ruelle} is due to the algorithmic approach developed by
Barvinok \cite[Section~2]{Bar16}.  It roughly states that if a polynomial $P(z)=\sum_{i=1}^nc_i z^i$
of degree $n$ is zero-free in a strip containing $[0,1]$, then $P(1)$ can be
$(1\pm\eps)$-approximated using $c_0,\dots,c_k$ for some $k=O\tp{\log \frac{n}{\eps}}$.

The basic idea is to truncate the Taylor expansion of $\log P(z)$ at $z=0$.  Let
$g(z)\defeq\log P(z)$ and for $k\ge 0$,
\begin{align*}
  T_k(g)(z)\defeq \sum_{i=0}^k\frac{g^{(i)}(0)}{i!} z^i,
\end{align*}
where $g^{(i)}$ is the $i$-th derivative of $g$.  In other words, $T_k(g)(z)$ is the first $k+1$
terms of the Taylor expansion of $g(z)$ at the origin.  Then \cite[Lemma 2.2.1]{Bar16} states the
following.

\begin{proposition}\label{prop:disk}
  Let $P(z)=\sum_{i=0}^n c_i z^i$ be a polynomial such that for some $\beta>1$, $P(z)$ is zero-free
  in the disk of radius $\beta$ centered at the origin.  Then there exists a constant $C_{\beta}$
  such that for any $0<\eps<1$,
  \begin{align*}
    \abs{\frac{\exp(T_k(g)(1))}{P(1)}-1}\le\eps,
  \end{align*}
  where $k=C_{\beta}\log\frac{n}{\eps}$.
\end{proposition}


This result states that we can approximately evaluate $P(1)$ using the first
$O\tp{\log\frac{n}{\eps}}$ terms of the Taylor expansion of $\log P(x)$ at the origin, when the
polynomial is zero-free in the disk of radius $\beta>1$.  If our polynomial $P_G(x)$ is zero-free in
the $\delta$-strip of $[0,1]$, then we can apply a transformation, \cite[Lemma 2.2.3]{Bar16},
to transform it into a polynomial that is zero-free in the disk of radius $>1$.

The following lemma describe the construction.

\begin{lemma}\label{lem:trans}
  Let $0<\delta<1$ be a constant and
  $\beta=1+\frac{\exp\tp{-\frac{1}{\delta}}}{2-2\exp\tp{-\frac{1}{\delta}}}>1$. There exists a
  polynomial $\phi_{\delta}(z)$ of degree $\exp\tp{O\tp{\frac{1}{\delta}}}$ such that
  \begin{itemize}
  \item [(1)] $\phi_\delta(0)=0$ and $\phi_\delta(1)=1$;
  \item [(2)] for every $z\in\mathbb{C}$ with $\abs{z}\le \beta$, the value $\phi_\delta(z)$ is
    within the $2\delta$-strip of $[0,1]$.
  \end{itemize}
\end{lemma}
\begin{proof}
  The idea to construct the polynomial $\phi_\delta$ is to start with the function $\log(z)$ (the
  principal logarithm) by noting that the logarithm function maps a circle centered at zero to an
  interval orthogonal to the real axis. We can then scale and shift the function to restrict the
  interval to some desired region. Finally, we construct the polynomial $\phi_\delta$ to approximate
  it.

  To this end, we let $h(z)\defeq \delta \log \frac{1}{1-\alpha z}$ where $\alpha$ is a parameter to
  be set. The condition $h(0)=0$ is automatically satisfied. To satisfy $h(1)=1$, we set
  $\alpha=1-\exp\tp{-\frac{1}{\delta}}$.
  Then $\beta=1+\frac{\exp\tp{-\frac{1}{\delta}}}{2-2\exp\tp{-\frac{1}{\delta}}}=\frac{1+\alpha}{2\alpha}$.
  It is easy to verify that for every $z\in\mathbb{C}$ with $\abs{z}\le \beta$, it holds that
  \[
    -\delta\log 2 \le \Re h(z)\le 1+\delta\log 2,
  \]
  and
  \[
    \abs{\Im h(z)}\le \frac{\pi}{2}\cdot\delta.
  \]
  We use a polynomial, namely the Taylor expansion of $h(z)$ at the origin to approximate
  $h(z)$. For every $k\ge 0$, the first $k$ terms of the Taylor expansion of $h$ at the origin is
  \[
    T_k(h)(z)=\delta \sum_{i=1}^k\frac{\alpha^i}{i}\cdot z^i.
  \]
  Then for
  $m=\frac{\log\tp{10(1+\alpha)}-\log\tp{1-\alpha}}{\log
    2-\log\tp{1+\alpha}}=\exp\tp{O\tp{\frac{1}{\delta}}}$, we have
  \[
    \abs{h(z)-T_k(z)}=\abs{\delta \sum_{i=m+1}^\infty\frac{\alpha^i}{i}\cdot z^i}\le
    \frac{2\delta}{(1-\alpha)(m+1)}\tp{\frac{1+\alpha}{2}}^{m+1}\le \frac{\delta}{10}.
  \]
  In particular, we have
  \[
    \abs{T_m(h)(1)-1}=\abs{T_m(h)(1)-h(1)}\le\frac{\delta}{10}.
  \]
  Finally, we define
  \[
    \phi_\delta(z)=\frac{T_m(h)(z)}{T_m(h)(1)}
  \]
  to force $\phi_\delta(1)=1$. This finishes the construction.
\end{proof}

Therefore, for a polynomial $P(z)$ that is zero-free in the $\delta$-strip of $[0,1]$, we can
use Proposition~\ref{prop:disk} to approximately evaluate
$P_\phi(z)\defeq P(\phi_{\frac{\delta}{2}}(z))$, which is zero-free in the disk of radius $\beta$ at
the origin for the value $\beta$ defined in Lemma~\ref{lem:trans}.
Note that $P(\phi_{\frac{\delta}{2}}(1))=P(1)$.

\begin{proposition}\label{prop:strip}
  Let $P(z)$ be a polynomial of degree $n$ such that for some $\delta>0$, $P(z)$ is zero-free in the
  $\delta$-strip of $[0,1]$. Then there exists a constant $C_\delta$ such that for any
  $0<\eps<1$,
  \[
    \abs{\frac{\exp\tp{T_k\tp{\log P_\phi}(1)}}{P(1)}-1}\le \eps,
  \]
  where $k=C_\delta\log\frac{n}{\eps}$.
\end{proposition}

At last, we show the Taylor expansion $T_k\tp{\log P_\phi}(1)$ can be computed efficiently from the coefficients of $P$.

\begin{proposition}\label{prop:P-phi-x}
  Let $P(z)$ be a polynomial of degree $n$ such that for some constant $\delta>0$, $P(z)$ is zero-free in the
  $\delta$-strip of $[0,1]$.
  For every $0\le k\le n$, assume that that we have oracle access to
  the first $k$ coefficients of $P(z)$, we can compute
  \[
    T_k\tp{\log P_\phi}(1)
  \]
  in time $O(k^2)$.
\end{proposition}

Since the degree of $\phi_{\frac{\delta}{2}}(z)$ is $\exp\tp{O\tp{\frac{1}{\delta}}}$, we can write
$P_{\phi}(z)=\sum_{i=1}^mc_i z^i$ where $m=n+C_\delta$ for some constant $C_{\delta}$ depending only on
$\delta$. It is easy to compute the coefficients $c_k$ given the coefficients of $P(z)$ of degree at
most $k$ in $O(k)$ time. Let $g_\phi\defeq \log P_\phi$, we now show how to compute $T_k(g_\phi)$ using
$(c_i)_{i\le k}$.

Let $z_1,\dots,z_m$ be the zeros of a polynomial $P_\phi(z)$ and for $0\le k\le m$, let
$p_k:=\sum_{i=1}^m z_i^{-k}$ be the $k$-th \emph{inverse power sum} of the zeros of $P_\phi(z)$.

Newton's identities state the relation between $(p_k)_k$ and the coefficients $(c_i)_i$.

\begin{proposition}[Newton's Identity]\label{prop:newton}
  For every $1\le k\le m$, it holds that
  \begin{align*}
    k\cdot c_k=-\sum_{i=0}^{k-1} c_i\cdot p_{k-i}
  \end{align*}
\end{proposition}

Newton's identities essentially provide a way to compute all $p_k$ consecutively using $(c_i)_i$,
and vice versa.  To be specific,
\begin{align}
  p_0 & = m; \notag\\
  p_k & = -c_0^{-1}\cdot\tp{\sum_{i=1}^{k-1}p_i\cdot c_{k-i}+k\cdot c_k} \mbox{ for $1\le k\le m$.} \notag
\end{align}
Therefore, it costs $O(k^2)$ time to compute $p_k$ using above recurrence.

On the other hand, we can write $P_\phi(z)=c_m\prod_{i=1}^m(z-z_i)$.
Recall that $g_\phi(z)=\log P_{\phi}(z)=\log
c_m+\sum_{i=1}^m\log\tp{z-z_i}$.

It is easy to calculate that for any $i\ge 1$,
\begin{align*}
  g_\phi^{(i)}(0) = - (i-1)!\sum_{j=1}^m z_j^{-i} = - (i-1)!p_i.
\end{align*}
Therefore,
\begin{align} \label{eqn:Taylor-p_i}
  T_k(g_\phi)(z)\defeq \log c_0 - \sum_{i=1}^k\frac{p_i}{i} z^i.
\end{align}
This proves \Cref{prop:P-phi-x}.

\subsection{Computing the inverse power sums}

Given \Cref{prop:disk} and \eqref{eqn:Taylor-p_i}, the main task then reduces to compute the first
$k$ inverse power sums $(p_i)_{i\le k}$.  We follow the method of Patel and Regts \cite{PR17a}.

We need some notations first.  Let $\+G$ be a family of all graphs, and $\+G_k$ be all graphs with
at most $k$ vertices.  We call a function $g:\+G\to\mathbb{C}$ a \emph{graph invariant} if
$g(G)=g(H)$ whenever $G\simeq H$.  A graph \emph{polynomial} is a graph invariant
$Q:\+G\rightarrow \=C[z]$, where $\=C[z]$ is the polynomial ring over $\=C$.  We call a graph
invariant $g(\cdot)$ \emph{additive} if for any two graphs $G$ and $H$, it holds that
$g(G\sqcup H)=g(G)+g(H)$, where $G\sqcup H$ is the graph consisting of disjoint copies of $G$ and
$H$.  Similarly, we call it \emph{multiplicative} if for every two graphs $G$ and $H$, it holds that
$g(G\sqcup H)=g(G)\cdot g(H)$.  For graphs $H$ and $G$, we use $\ind{H}{G}$ to denote the number of
\emph{induced} subgraphs of $G$ isomorphic to $H$.  Then $\ind{H}{\cdot}$ is a graph invariant for a
fixed graph $H$.  By convention let $\ind{\emptyset}{G}=1$ for any $G$.

\begin{definition}\label{def:BIGCP}
  Let $Q(G)(z)=\sum_{i=1}^{d(G)}a_i(G) z^{i}$ be a multiplicative graph polynomial of degree $d(G)$
  such that $Q(G)(0)=1$ for any $G$.  We call $Q(\cdot)$ a \emph{bounded induced graph counting
    polynomial (BIGCP)} if there are constants $\alpha,\beta\in\=N$ such that the following holds:
  \begin{itemize}
  \item for every graph $G$, there exist $\lambda_{H,i}\in\=C$ such that
    \begin{align} \label{eqn:BIGCP-induce} a_i(G)=\sum_{H\in\+G_{\alpha i}}\lambda_{H,i}\cdot
      \ind{H}{G};
    \end{align}
  \item for every $H\in\+G_{\alpha i}$, $\lambda_{H,i}$ can be computed in time
    $\exp\tp{\beta \cdot\abs{V(H)}}$, where $V(H)$ is the set of vertices of $H$.
  \end{itemize}
\end{definition}

Patel and Regts \cite[Theorem 3.2]{PR17a} has shown that the inverse power sums can be computed for
BIGCP in single exponential time.

\begin{proposition}\label{prop:BIGCP-count}
  Let $\Delta\in\=N$, $G$ be a graph with
  maximum degree $\Delta$ and $Q(G)(\cdot)$ be a BIGCP.  There is a deterministic $\exp\tp{C\Delta k}$-time algorithm, which computes
  the inverse power sums $(p_i)_{i\le k}$ of $Q(G)(\cdot)$, for some constant $C>0$.
\end{proposition}

To our need, we just need to verify that $P_G(\cdot)$ from \eqref{eqn:Holant-polynomial} is a BIGCP,
whenever $f_0=1$.

\begin{lemma} \label{lem:Holant-BIGCP} Let $G=(V,E)$ be a $\Delta$-regular graph and $f=\sqtp{f_0,f_1,\ldots,f_\Delta}$ be a signature. If $f_0=1$, then the Holant polynomial $P_G(\cdot)$ is a BIGCP with $\alpha =2$ and $\beta = C \Delta$ for some constant $C>0$.
\end{lemma}
\begin{proof}
  Clearly $P_G(0)=Z_0(G)=f_0^{\abs{V}}=1$. We would like to define $\lambda_{H,i}$ so that for every
  $1\le i\le n$,
  \begin{align}\label{eq:Ziind}
    Z_i(G)=\sum_{H\in\+G_{2i}}\lambda_{H,i}\cdot\ind{H}{G}.
  \end{align}

  For any $\sigma\in \set{0,1}^E$, let $G[\sigma]$ be the subgraph induced by the set of vertices
  with at least $1$ adjacent edges under $\sigma$.  Let $S_i$ be the set of subgraphs induced by
  assignments of Hamming weight $i$, namely
  $S_i\defeq\set{G[\sigma]\cmid \sigma\in \set{0,1}^E\mbox{ and }\abs{\sigma}=i}$.  The equivalence
  relation of graph isomorphisms induces a partition of $S_i$.  We choose one graph from each
  equivalence class and denote this family of graphs by $\+H_i$.  Therefore, for every two distinct
  graphs $H_1,H_2\in \+H_i$, they are not isomorphic.  Moreover, as $G[\sigma]$ has at most $2i$
  vertices, $\+H_i\subseteq \+G_{2i}$.

  For every $H\in\+H_i$, consider an assignment $\pi$ of signatures, where $v\in V$ of degree
  $d\le\Delta$ is assigned $[f_0,f_1,\dots,f_d]$, a truncated $f$.  Let
  \begin{align*}
    \lambda_{H,i}\defeq Z_i(H;\pi).
  \end{align*}
  To verify \eqref{eq:Ziind}, we rewrite
  \begin{align*}
    Z_i(G)
    &=\sum_{\sigma\in\set{0,1}^E\text{ and }\abs{\sigma}=i}\;\prod_{v\in V} f(\sigma|_{E(v)})\\
    &=\sum_{H\in \+G_{2i}}\;\sum_{\substack{\sigma\in\set{0,1}^E\\\abs{\sigma}=i\text{ and }G[\sigma]\simeq H}}\;\prod_{v\in V} f(\sigma|_{E(v)})\\
    &=\sum_{H\in \+G_{2i}}\;
      \sum_{\substack{G'\text{ is an induced subgraph of $G$}\\G'\simeq H}}\;
    \sum_{\substack{\sigma\in\set{0,1}^E\\\abs{\sigma}=i\text{ and }G[\sigma]=G'}}\;\prod_{v\in V} f(\sigma|_{E(v)})\\
    &=\sum_{H\in \+G_{2i}}\;
      \sum_{\substack{G'\text{ is an induced subgraph of $G$}\\G'\simeq H}}\; Z_i(G';\pi)\cdot f_0^{\abs{V\setminus V(H)}}\\
    &=\sum_{H\in \+G_{2i}}\; Z_i(H;\pi)\cdot \ind{H}{G},
  \end{align*}
  since $Z_i(G';\pi)=Z_i(H;\pi)$ whenever $G'\simeq H$.  Thus \eqref{eq:Ziind} holds.

  Since $\+H_i\subseteq \+G_{2i}$, we have that $\alpha =2$.  Moreover, $H$ contains at most
  $\Delta\abs{V(H)}$ edges.  As a consequence, $Z_i(H;\pi)$ can be computed in time
  $2^{O(\Delta \abs{H})}$.  Thus, we can take $\beta=C \Delta$ for some constant~$C>0$.
\end{proof}

Gathering what we have seen so far, we have the following theorem.

\begin{theorem} \label{thm:roots-Holant}
  Let $f$ be a symmetric signature of arity $\Delta$.  If the
  local polynomial $P_f(x)$ is $H_{\eps}$-stable for some $\eps>0$, then there is an FPTAS for
  $\Holant(f)$.
\end{theorem}
\begin{proof}
  Since $P_f(x)$ is $H_{\eps}$-stable, $f_0\neq 0$.  We may thus normalize $f$ so that $f_0=1$.  By
  \Cref{thm:Ruelle}, $P_f(x)$ being $H_{\eps}$-stable implies that for any $\Delta$-regular
  $G=(V,E)$, $P_G(x)$ is zero-free in a $\delta$-strip containing $[0,1]$.  Recall that
  $Z(G;f)=P_G(1)$.  By \Cref{prop:newton}, we can
  $(1\pm\eps)$-approximate $P_G(1)$ using $\exp(T_k(\log P_G)(x))$ for some
  $k=O\tp{\log \frac{m}{\eps}}$, where $m=\abs{E}$.  In order to compute $T_k(\log P_G)(x)$, we use
  \Cref{prop:BIGCP-count} and \Cref{lem:Holant-BIGCP} to compute the inverse power sums $(p_i)$ of $P_G(x)$, and
  then apply \Cref{prop:newton} to get the first $k$ coefficients of $P_G(x)$.
  The theorem then follows from \Cref{prop:P-phi-x}.
\end{proof}

\begin{remark}
  \Cref{thm:roots-Holant} is a sufficient but not necessary condition for a $\Holant$ problem to be approximable.
  To see this, once again, consider the problem of counting even subgraphs.
\end{remark}

\section{Holographic transformations} \label{sec:holotrans}

\Cref{thm:roots-Holant} implies an FPTAS for $\Holant(f)$ if $f$ is $H_{\eps}$-stable.
However, an FPTAS may still exist even if $f$ is not $H_{\eps}$-stable.
One way to extend the reach of this approach is via Valiant's holographic transformation \cite{Val08},
which changes $f$ but preserves the partition function.
We remark that even with holographic transformations, this approach is not exhaustive.
An example is the problem of counting even subgraphs.

We use $\holant{f}{g}$ to denote the Holant problem where the input is a bipartite graph $H = (U,V,E)$.
Each vertex in $U$ or $V$ is assigned the signature $f$ or $g$, respectively.
Call this assignment $\pi$, namely $\pi(u)=f$ for any $u\in U$ and $\pi(v)=g$ for any $v\in V$.
Recall \eqref{eqn:Holant-def}, and $Z(H;\pi)$ is the output of the computational problem $\holant{f}{g}$.
The signature $f$ is considered as a row vector (or covariant tensor),
whereas the signature $g$ is considered as a column vector (or contravariant tensor).

Let $T$ be an invertible $2$-by-$2$ matrix.
Let $d_1=\arity(f)$ and $d_2=\arity(g)$.
Define $f'=f\cdot T^{\otimes d_1}$ and $g'=\tp{T^{-1}}^{\otimes d_2} g$.
Let $\pi'$ be the assignment such that $\pi'(u)=f'$ for any $u\in U$ and $\pi'(v)=g'$ for any $v\in V$.

\begin{proposition}[Valiant's Holant Theorem~\cite{Val08}]\label{prop:holographic-transformation}
 If $T \in \=C^{2 \times 2}$ is an invertible matrix,
 then for any bipartite graph $H$, $Z(H;\pi) = Z(H;\pi')$,
 where $\pi'$ is defined above.
\end{proposition}

Therefore, an invertible holographic transformation does not change the complexity of the Holant problem in the bipartite setting.
For a (non-bipartite) Holant problem, we can always view the edge as a binary equality function $=_2$.
Thus, $\Holant(f)$ is the same as $\holant{f}{=_2}$.
Let $\*O_2(\=C)$ be the set of $2$-by-$2$ orthogonal matrices,
namely $\*O_2(\=C)=\set{T\in\=C^{2\times 2}\mid T\transpose{T}=I_2}$.
As orthogonal transformations preserve the binary equality,
the following result will become handy in the standard setting.

\begin{proposition}[\cite{CLX11}] \label{prop:orthogonal}
 If $T \in \*O_2(\=C)$ is an orthogonal matrix 
 then for any $d$-regular graph $G$ and a signature $f$ of arity $d$, $Z(G;f) = Z(G;f\cdot T^{\otimes d})$.
\end{proposition}

As a particular consequence of \Cref{prop:orthogonal},
under the transformation $\trans{0}{1}{1}{0}$,
the complexity of $\Holant(f)$ is equivalent to $\Holant(\ol f)$ where $\ol f=[f_d,f_{d-1},\dots,f_0]$.
We will use this fact in the following without explicitly mentioning it.

\section{Second-order recurrences}\label{sec:2nd-rec}

The aim of this section is to study the locations of zeros of local polynomials of signatures
satisfying generalised second-order recurrences in order to apply \Cref{thm:roots-Holant}. Specifically, we identify the family of signatures
whose local polynomials are $H_{\eps}$-stable for some $\eps>0$, under some suitable holographic transformations.

For a tuple of reals $(a,b,c)\ne (0,0,0)$, define
\[
  \+F_{a,b,c} \defeq \set{ [f_0,f_1,\dots,f_{d}] \cmid a f_k + b f_{k+1} + c f_{k+2} = 0, \forall 0\le k\le d-2, \mbox{ and
    } f_k\ge 0, \forall 0\le k \le d}.
\]
The family $\+F_{a,b,c}$ consists of signatures with non-negative entries satisfying second-order
linear recurrence relation parameterized by $(a,b,c)$.  Whenever $\+F_{a,b,c}$ appears, we always
assume that $(a,b,c)\ne (0,0,0)$.


The following proposition states the general form of a function satisfying a generalised
second-order recurrence.

\begin{proposition}\label{prop:2order}
  Let $f=[f_0,\dots,f_d]\in \+{F}_{a,b,c}$ be a signature and $c\ne 0$.  There are two cases:
  \begin{itemize}
  \item if $b^2\ne 4ac$, then
    \[
      f_k= x\phi_1^k+y\phi_2^k,
    \]
    where $\phi_1,\phi_2$ are the two roots of the polynomial $cz^2+bz+a=0$ and $x,y$ are two
    constants determined by $f_0$ and $f_1$;
  \item if $b^2=4ac$, then
    \[
      f_k=x\phi^k+yk\phi^{k-1},
    \]
    where $\phi$ is the unique root of the polynomial $cz^2+bz+a=0$ and $x,y$ are two constants
    determined by $f_0$ and $f_1$. In case of $\phi=0$, we follow the convention that $0\cdot 0^{-1}=0$.
  \end{itemize}
\end{proposition}



In this section, we assume that all signatures (or their reversals) in consideration have nonzero leading term,
i.e., $f_0\ne 0$ or $f_d\ne 0$. We will discuss the case of $f_0=f_d=0$ in
Section~\ref{sec:exception}.

We will use $\mathcal{F}^*_{a,b,c}$ to denote the subset family of $\mathcal{F}_{a,b,c}$ with this
additional property $f_0>0$. It turns out that the behaviour of
signatures in $\+F^*_{a,b,c}$ is closely related to the sign of the value $b^2-4ac$, namely the discriminant of the characteristic polynomial
$cz^2+bz+a$. Therefore, our discussion is divided into three parts.





\subsection{\texorpdfstring{$b^2-4ac>0$}{b2-4ac>0}}

In this case, the characteristic polynomial of signatures in $\+F^*_{a,b,c}$ has two distinct real
roots.  We first single out a special case.

\begin{lemma} \label{lem:ps+qt<0}
  Let $f$ be a symmetric signature of arity $d\ge 3$, where $d$ is an odd integer,
  $f_i\ge 0$ for all $i=0,1,\dots,d$, and $f$ is not identically zero.
  If there exist $p,q,s,t\in\=R$ such that
  $p^2+q^2=s^2+t^2$, $ps+qt<0$, and $f=\tp{p,q}^{\otimes d}+\tp{s,t}^{\otimes d}$, then up to a
  non-zero scaler, $f$ or $\ol f$ is $[1,0,\lambda^2,0,\dots,\lambda^{d-1},0]$ for some
  $\lambda> 1$, where $\ol f=\sqtp{f_d,f_{d-1},\ldots,f_0}$
\end{lemma}
\begin{proof}
%
  Since $f=\tp{p,q}^{\otimes d}+\tp{s,t}^{\otimes d}$,
  we have $f_i=q^ip^{d-i}+t^is^{d-i}$.
  We discuss the sign of $qt$.

  First assume $qt\ge 0$.
  The fact $f_1\ge 0$ yields
  \[
    qp^{d-1}+ts^{d-1}\ge 0.
  \]
  Since $d$ is odd, then $q$ and $t$ must be both non-negative. Let $t=\sqrt{p^2+q^2-s^2}\ge 0$. It
  follows from $ps+qt<0$ that $ps<0$. We can assume without loss of generality that $p>0$,
  $s< 0$ and $\abs{p}\ge \abs{s}$ (a consequence of $f_0\ge 0$). To ease the presentation, let
  $s'=-s > 0$. Then
  \begin{align*}
    ps+qt < 0 \iff qt < ps' \iff q^2(p^2+q^2-s'^2) < p^2s'^2 \iff \abs{q} < \abs{s'}.
  \end{align*}
  We then consider the requirement $f_{d-1}\ge 0$. This is equivalent to
  \begin{align*}
    q^{d-1}p+t^{d-1}s\ge 0
    &\iff q^{d-1}p\ge t^{d-1}s'\\
    &\iff q^2p^{\frac{2}{d-1}}\ge (p^2+q^2-s'^2)s'^{\frac{2}{d-1}}\\
    &\iff q^2(p^{\frac{2}{d-1}}-s'^{\frac{2}{d-1}})\ge (p^2-s'^2)s'^{\frac{2}{d-1}}.
  \end{align*}
  We apply $\abs{q} < \abs{s'}$ and obtain
  \begin{align*}
    (p^2-s'^2)s'^{\frac{2}{d-1}}\le s'^2(p^{\frac{2}{d-1}}-s'^{\frac{2}{d-1}})
    & \iff \frac{p^2}{s'^2}-1\le \frac{p^{\frac{2}{d-1}}}{s'^{\frac{2}{d-1}}}-1\\
    & \iff \abs{s}\ge \abs{p}.
  \end{align*}
  Therefore, it must hold that $p=-s$, $q=t$ and we have
  $f=(p,q)^{\otimes d}+(-p,q)^{\otimes d}$.
  Moreover, $ps+qt<0$ implies that $p>q$.
  If $q=t=0$, then $f$ is identically zero, a contradiction.
  Otherwise $q>0$, and we can choose $\lambda=\frac{p}{q}>1$ and $\ol f$ is
  $[1,0,\lambda,0,\lambda^2,0,\dots]$ up to a non-zero scalar.

  Now we assume $qt<0$, and without loss of generality further assume that $q > 0$ and $t < 0$. Let
  $t=-\sqrt{p^2+q^2-s^2}$.  We distinguish between $ps\ge 0$ and $ps<0$.
  \begin{itemize}
  \item[(i)] If $ps\ge 0$, then
    \begin{align}
      ps<-qt
      &\implies p^2s^2<q^2(p^2+q^2-s^2)\notag\\
      &\implies \abs{s}<\abs{q}.      \label{eq:sq1}
    \end{align}
    Again, $f_1\ge 0$ implies that $qp^{d-1}+ts^{d-1}\ge 0$. This is equivalent to
    \begin{align*}
      qp^{d-1}\ge\sqrt{p^2+q^2-s^2}\cdot s^{d-1}
      &\iff q^2p^{2d-2}\ge (p^2+q^2-s^2)\cdot s^{2d-2}\\
      &\iff q^2(p^{2d-2}-s^{2d-2})\ge s^{2d-2}(p^2-s^2).
    \end{align*}
    Since $f_d\ge 0$, we have $\abs{q}\ge\abs{t}$.
    Together with $p^2+q^2=s^2+t^2$, it implies that $p^2\le s^2$.
    Thus we have either $\abs{p}=\abs{s}$ or
    \[
      q^2\le \frac{s^{2d-2}(s^2-p^2)}{s^{2d-2}-p^{2d-2}}.
    \]

    If $\abs{p}=\abs{s}$, then $p=s\ge 0$ and $t=-q$.
    In this case, $f=(p,q)^{\otimes d}+(p,-q)^{\otimes d}$.
    If $p=0$, then $f$ is identically zero, a contradiction.
    Thus $p>0$, and we can choose $\lambda=\frac{q}{p}$, and $\lambda>1$ because $0>ps+qt=p^2-q^2$.

    Otherwise, since
    \[
      s^{2d-2}-p^{2d-2}=(s^2-p^2)\tp{\sum_{i=0}^{d-2}s^{2i}p^{2(d-2-i)}}\ge s^{2d-4}(s^2-p^2),
    \]
    we have $q^2\le s^2$.
    This contradicts to \eqref{eq:sq1}.
  \item[(ii)] If $ps < 0$, we first assume that $p < 0$ and $s > 0$. In this case, we let $p'=-p$
    and $t'=-t$. Then $f_0,f_1,f_2\ge 0$ implies
    \[
      s^d\ge p'^d;\quad t's^{d-1}\le q p'^{d-1};\quad t'^2s^{d-2}\ge q^2p'^{d-2},
    \]
    where $p',q,t',s$ above are all positive.
    The first two imply that $t'p'\le qs$,
    and the last two imply that $t'p'\ge qs$.
    Thus $t'p' = qs$.
    This is further equivalent to $s^2q^2=p^2(p^2+q^2-s^2)$, or $(p^2+q^2)(p^2-s^2)=0$.
    It implies that either $p=q=0$ or $p=-s$.
    In both cases, $f$ is identically zero, a contradiction.

    Finally, consider the case when $p > 0$ and $s < 0$. Then $f_0=p^d+s^d\ge 0$ implies
    $\abs{p}\ge\abs{s}$. On the other hand, $f_d=q^d+t^d\ge 0$ is equivalent to $\abs{q}\ge \abs{t}$.
    However $p^2+q^2=s^2+t^2$.
    Thus we have $p=-s$ and $q=-t$. This means that $f$ is identically zero, also a contradiction. \qedhere
  \end{itemize}
\end{proof}

Let $=_d$ be the equality function of arity $d$, namely the function $[1,0,\dots,0,1]$.
We call the problem
$\holant{=_d}{[\beta,1,\beta]}$ a \emph{ferromagnetic Ising model without external fields}, if
$\beta>1$.  An FPRAS for this problem has been given by Jerrum and Sinclair \cite{JS93}.  Then we
have the following lemma.

\begin{lemma}\label{lem:delta1}
  Let $f=\sqtp{f_0,f_1,\dots,f_d}\in \+F^*_{a,b,c}$ with $b^2-4ac>0$. Then one of the following
  holds:
  \begin{itemize}
  \item $\Holant\tp{f}$ can be solved exactly in polynomial-time; or
  \item there is an invertible matrix $M\in\=C^{2\times 2}$ such that
    $\holant{fM^{\otimes d}}{\tp{M^{-1}}^{\otimes 2}\cdot\tp{=_2}}$ is a ferromagnetic Ising model
    without external fields; or
  \item there is an orthogonal matrix $M\in \*O_2(C)$ such that either
    $P_{f\cdot M^{\otimes d}}\tp{z}$ or $P_{\ol f\cdot M^{\otimes d}}\tp{z}$ is
    $H_{\varepsilon}$-stable for some $\varepsilon>0$, where $\ol f=\sqtp{f_d,f_{d-1},\ldots,f_0}$; or
  \item $f$ or $\ol f$ is $[1,0,\lambda^2,0,\lambda^4,0,\dots,\lambda^{d-1},0]$ for some $\lambda > 1$ and has an odd arity $d$.
  \end{itemize}
\end{lemma}

\begin{proof}
  If $c=0$, then $af_k+b f_{k+1}=0$ for all $k\le d-2$.  Thus, $f_0,\dots,f_{d-1}$ form a geometric
  sequence with some ratio $\phi\in\=R$, and $f$ can be written as
  $f=x\tp{1,\phi}^{\otimes d}+y\tp{0,1}^{\otimes d}$, where $x,y,\phi\in\=R$.  Pulling $x$ and $y$
  into the tensor power, there exist $p,q,s,t\in\=R$ and $r=1$ or $-1$ such that $f$ is a non-zero
  multiple of $\tp{p,q}^{\otimes d} + r\tp{s,t}^{\otimes d}$.

  Otherwise $c\neq 0$.  It follows from \Cref{prop:2order} that we can rewrite
  $f=x\tp{1,\phi_1}^{\otimes d}+y\tp{1,\phi_2}^{\otimes d}$, where $\phi_1,\phi_2\in\=R$ and
  $\phi_1\neq\phi_2$.  Since $f$ has non-negative weights, it implies that $x,y\in\=R$ as well.
  Thus, similar to the case above, there exist $p,q,s,t\in\=R$ and $r=1$ or $-1$ such that $f$ is a
  non-zero multiple of $\tp{p,q}^{\otimes d} + r\tp{s,t}^{\otimes d}$.

  The four possibilities of the lemma come from the values these reals might take. If $pt=qs$, then
  $f$ is degenerate and the partition function can be computed in polynomial time (see
  e.g.~\cite[Chapter 2]{CC17}).  Thus we assume $pt-qs\ne 0$ in the following.

  First we consider the case that $p^2+q^2=s^2+t^2$. We claim that we can always write
  $f=(p,q)^{\otimes d}+(s,t)^{\otimes d}$ without loss of generality. To see this, we distinguish
  between the parity of $d$. If $d$ is odd, then
  $(p,q)^{\otimes d}-(s,t)^{\otimes d} = (p,q)^{\otimes d}+(-s,-t)^{\otimes d}$. If $d$ is even, we
  know from $f=(p,q)^{\otimes d}-(s,t)^{\otimes d}$ that $f_0= p^d-s^d$ and
  $f_d=q^d-t^d$. Therefore, $f_0>0$ and $f_d\ge 0$ imply $p^2>s^2$ and $q^2\ge t^2$, which
  contradicts $p^2+q^2=s^2+t^2$.

  We write $\Holant\tp{f}$ as $\holant{f}{=_2}$. Let $M'=\trans{p}{q}{s}{t}$ be an invertible matrix
  and $M=M'^{-1}$. It follows from~\Cref{prop:holographic-transformation} that $\holant{f}{=_2}$ is
  equivalent to $\holant{fM^{\otimes d}}{\tp{M^{-1}}^{\otimes 2}\cdot\tp{=_2}}$. We verify that this
  particular Holant problem is either solvable in polynomial-time or equivalent to a ferromagnetic
  Ising model without external fields. We have
  \begin{align*}
    fM^{\otimes d} = \tp{\tp{1,0}^{\otimes d} + \tp{0,1}^{\otimes d}}M'^{\otimes d}M^{\otimes} = \tp{1,0}^{\otimes d} + \tp{0,1}^{\otimes d},
  \end{align*}
  and
  \begin{align*}
    \tp{M^{-1}}^{\otimes 2}\cdot\tp{=_2}
    = M'^{\otimes 2}\cdot\tp{=_2} = \transpose{\tp{p^2+q^2,ps+qt,ps+qt,s^2+t^2}}.
  \end{align*}
  If $ps+qt=0$, clearly it is solvable in polynomial-time since the edges in every component of the
  instance must be assigned with the same value in order to contribute a non-zero weight to the
  partition function. If $ps+qt> 0$, we have that
  $\tp{p^2+q^2}\tp{s^2+t^2}-\tp{ps+qt}^2=\tp{pt-qs}^2>0$, and it is a ferromagnetic Ising model
  without external fields.
  If $ps+qt<0$ and $d$ is even,
  then a further transformation $\trans{1}{0}{0}{-1}$ makes the middle term positive,
  and it is a ferromagnetic Ising model again.
  Otherwise, \Cref{lem:ps+qt<0} applies, and we are in the last case of the lemma.

  The remaining case is that $pt\ne qs$ and $p^2+q^2\neq s^2+t^2$.  If $\abs{q}=\abs{t}$, then
  $\abs{p}\neq \abs{s}$ and we replace $\tp{p,q,s,t}$ by $\tp{q,p,t,s}$. This is equivalent to work
  with $\ol f$. So from now on we also assume that $\abs{q}\neq \abs{t}$.  Let
  $M'=\trans{w}{1}{1}{-w}$ where $w\in\=R$ is a parameter to be set later.
  Then $f\cdot M'^{\otimes d}$ is $\tp{q+pw,p-qw}^{\otimes d}+a\tp{t+sw,s-tw}^{\otimes d}$ and
  \begin{align*}
    P_{f\cdot M'^{\otimes d}}\tp{z}=\tp{q+pw+\tp{p-qw}z}^d + a\tp{t+sw+\tp{s-tw}z}^d.
  \end{align*}
  The zeros of this polynomial must satisfy
  \begin{align}
    \abs{q+pw+\tp{p-qw}z}=\abs{t+sw+\tp{s-tw}z}. \label{eq:roots1}
  \end{align}
  We show that by choosing appropriate $w$ the roots to this equation are in the open left
  half-plane.

  If $p=q=0$ and $s-tw\neq 0$, the roots to the equation \eqref{eq:roots1} must be
  $-\frac{t+sw}{s-tw}$. Since $p^2+q^2\neq s^2+t^2$, it holds that $\tp{s,t}\neq \tp{0,0}$.  There
  are four cases.
  \begin{itemize}
  \item If $t=0$, let $w=1$. It holds that $s-tw=s\neq 0$ and $-\frac{t+sw}{s-tw}=-w<0$.
  \item If $s=0$, let $w=-1$. It holds that $s-tw=t\neq 0$ and $-\frac{t+sw}{s-tw}=\frac{1}{w}<0$.
  \item If $st<0$, let $w=\frac{2s}{t}<0$. It holds that $s-tw=-s\neq 0$ and
    $-\frac{t+sw}{s-tw}=\frac{t}{s}+w<0$.
  \item If $st>0$, let $w=0$. It holds that $s-tw=s\neq 0$ and $-\frac{t+sw}{s-tw}=-\frac{t}{s}<0$.
  \end{itemize}
  The case of $s=t=0$ is completely analogous.

  Now we can make the further assumption that $\tp{p,q}\neq \tp{0,0}$ and $\tp{s,t}\neq \tp{0,0}$.
  Let $\alpha = -\frac{p-qw}{s-tw}\in\=R$ be another parameter, which eventually will be set to $1$
  or $-1$.  As $w=\frac{\alpha s+p}{\alpha t+q}$ and $\abs{q}\neq \abs{t}$, the value of the
  parameter $w$ will be determined when the sign of $\alpha$ is chosen.
  Since $p-qw=\frac{\alpha \tp{pt-qs}}{\alpha t + q}\neq 0$, we let $z_1 = -\frac{q+pw}{p-qw}$ which
  is well-defined.  Similarly it holds that $s-tw=\frac{qs-pt}{\alpha t + q}\neq 0$, and we let
  $z_2=-\frac{t+sw}{s-tw}$.  The equation \eqref{eq:roots1} is equivalent to
  \begin{align}
    \abs{\alpha}\cdot\abs{z-z_1} = \abs{z-z_2}. \label{eq:roots2}
  \end{align}
  Since $\abs{\alpha}=1$, in order to make the roots to the equation \eqref{eq:roots2} in the open
  left half-plane, it suffices to make sure that
  \begin{align}
    z_1 + z_2 = \frac{\tp{p^2+q^2}-\tp{s^2+t^2}}{\alpha\tp{qs-pt}} < 0. \label{eq:negative}
  \end{align}
  Since $p^2+q^2\neq s^2+t^2$, we can let $\alpha=-1$ if
  $\frac{\tp{p^2+q^2}-\tp{s^2+t^2}}{qs-pt} > 0$, or let $\alpha =1$ otherwise.

  We have showed that there is a matrix $M'\in \=C^{2\times 2}$ such that the zeros of
  $P_{f\cdot M'^{\otimes d}}\tp{z}$ are in the open left half-plane.  Since a polynomial has only a
  finite number of zeros, there is a constant $\varepsilon>0$ that $P_{f\cdot M'^{\otimes d}}\tp{z}$
  is $H_{\varepsilon}$-stable.  It holds that
  $M'\transpose{\tp{M'}}=\trans{1+w^2}{0}{0}{1+w^2}=\tp{1+w^2}I_2$ where $1+w^2>0$ as $w\in\=R$. Let
  $M=\frac{1}{\sqrt{1+w^2}}M'$. Clearly $M\transpose{M}=I_2$ and $M\in\*O_2(\=C)$.  Since
  $P_{f\cdot M'^{\otimes d}}\tp{z}=\tp{1+w^2}^{d/2} P_{f\cdot M^{\otimes d}}\tp{z}$,
  $P_{f\cdot M^{\otimes d}}\tp{z}$ has the same set of zeros as $P_{f\cdot M'^{\otimes d}}\tp{z}$.
  So $P_{f\cdot M^{\otimes d}}$ is also $H_\varepsilon$-stable for some $\varepsilon>0$.
\end{proof}

\subsection{\texorpdfstring{$b^2-4ac=0$}{b2-4ac=0}}

When the characteristic polynomial of $f$ has only one real root of
multiplicity two, we show that there always exists an orthogonal
transformation to reduce $f$ to a function whose local polynomial is
$H_\eps$-stable.

\begin{lemma} \label{lem:delta0} Let
  $f=\sqtp{f_0,f_1,\ldots,f_d}\in \+F^*_{a,b,c}$ with $b^2-4ac=0$,
  then there is an orthogonal matrix $M\in\*O_2(\=C)$ such that
  $P_{fM^{\otimes d}}\tp{z}$ is $H_{\varepsilon}$-stable for some
  $\varepsilon>0$.
\end{lemma}

\begin{proof}
  If $c=0$, then $b=0$ since $b^2-4ac=0$. This cannot happen because
  if so $f_0$ would be zero.  If $b=0$, then $a=0$ since $c\neq 0$ and
  $b^2-4ac=0$.  In this case, $f$ is of form
  $\sqtp{f_0,f_1,0,\ldots,0}$ and we can simply pick $M=I_2$.  Clearly
  $P_{fM^{\otimes d}}\tp{z}=f_0 + df_1 z$, which is
  $H_\varepsilon$-stable for some $\varepsilon>0$ since $f_0>0$ and
  $f_1\ge 0$.

  So now we assume that $b\neq 0$.  Since $c\neq 0$ and $b^2-4ac=0$,
  the equation $cz^2 + bz + a =0$ has one real root with multiplicity
  two and we denote it by $\phi$.  Note that
  $\phi=-\frac{b}{2c}\neq 0$ since $b\neq 0$.  It follows from
  \Cref{prop:2order} that $f_k=x\phi^k+y'k\phi^{k-1}$ for
  $0\le k \le d$ and some $x,y'\in \=R$. Since $\phi\ne 0$, to ease
  the presentation, we let $y=\frac{y'}{\phi}$ and rewrite
  $f_k=x\phi^k+yk\phi^k$. Clearly $x=f_0>0$. We can write $f$ as
  \begin{align*}
    f=x\tp{1,\phi}^{\otimes d} + y\sum_{k=1}^d \tp{1,\phi}^{\otimes \tp{k-1}}\otimes \tp{0,\phi} \otimes \tp{1,\phi}^{\otimes \tp{d-k}}.
  \end{align*}
  Let $M'=\trans{1}{w}{-w}{1}$ where $w\in\=R$ is a parameter to be
  set later. Then
  \begin{align*}
    f\cdot M'^{\otimes d}= x\tp{1-\phi w,\phi+w}^{\otimes d}
    +y\sum_{k=1}^d \tp{1-\phi w,\phi+w}^{\otimes\tp{k-1}}\otimes \tp{-\phi w,\phi}\otimes \tp{1-\phi w,\phi+w}^{\otimes \tp{d-k}},
  \end{align*}
  and
  \begin{align*}
    P_{f\cdot M'^{\otimes d}}\tp{z}=x\tp{1-\phi w+\tp{\phi+w}z}^d + yd\tp{1-\phi w+\tp{\phi+w}z}^{d-1}\tp{-\phi w+\phi z}.
  \end{align*}
  The zeros of this polynomial must satisfy
  \begin{align}
    \tp{1-\phi w+\tp{\phi+w}z}^{d-1} \tp{x-\tp{x+yd}\phi w+\tp{xw+\tp{x+yd}\phi}z}=0. \label{eq:zeros8}
  \end{align}
  If $\phi+w\neq 0$ and $xw+\tp{x+yd}\phi\neq 0$, then the roots of
  this equation must be of the form $-\frac{1-\phi w}{\phi+w}$ or
  $-\frac{x-\tp{x+yd}\phi w}{xw+\tp{x+yd}\phi}$.  We choose
  appropriate $w$ and check that these two roots are negative,
  $\phi+w\neq 0$ and $xw+\tp{x+yd}\phi\neq 0$.  Recall that
  $\phi\neq 0$ and $x=f_0 > 0$.  We discuss various cases depending on
  the sign of $\phi$ and $x+yd$.
  \begin{itemize}
  \item If $x+yd=0$, then the roots of the equation \eqref{eq:zeros8}
    are $-\frac{1-\phi w}{\phi+w}$ and $-\frac{1}{w}$. If $\phi<0$,
    let $w=-2\phi>0$ and
    $-\frac{1-\phi w}{\phi+w}=\frac{1+2\phi^2}{\phi}<0$. If $\phi>0$,
    let $w=\frac{1}{2\phi}>0$ and
    $-\frac{1-\phi
      w}{\phi+w}=-\frac{1}{2\phi+\frac{1}{\phi}}<0$. Clearly
    $\phi+w\neq 0$ and $xw+\tp{x+yd}\phi \neq 0$ in both cases.
  \item If $\phi>0$ and $x+yd>0$, then let
    $w=\min\set{\frac{1}{2\phi}, \frac{x}{2\tp{x+yd}\phi}}>0$. It
    holds that
    \begin{align*}
      -\frac{1-\phi w}{\phi+w} & \le -\frac{1}{2\tp{\phi+w}} < 0,\\
      -\frac{x-\tp{x+yd}\phi w}{xw+\tp{x+yd}\phi} &\le -\frac{x}{2\tp{xw+\tp{x+yd}\phi}} < 0.
    \end{align*}
    Whatever $w=\frac{1}{2\phi}$ or $w=\frac{x}{2\tp{x+yd}\phi}$, it
    is clear that $\phi+w\neq0$ and $xw+\tp{x+yd}\phi\neq 0$.
  \item If $\phi>0$ and $x+yd<0$, then $f_d = \phi^d\tp{x+yd} <
    0$. This contradicts to $f_d\ge 0$.
    \item If $\phi<0$ and $x+yd>0$, then consider $f_d = \phi^d\tp{x+yd}$.
      If $d$ is odd, then $f_d < 0$. Contradiction.
      Thus $d$ must be even.
      Then $\phi^{d-1}<0$. Since $f_{d-1} = \phi^{d-1}\tp{x+y\tp{d-1}}\ge 0$, it holds that $x+y\tp{d-1}\le 0$.
      As $x>0$, $y$ must be negative, and then it contradicts to $x+yd>0$.
    \item If $\phi<0$ and $x+yd<0$, then consider $f_d=\phi^d\tp{x+yd}$. If $d$ is even, then $f_d < 0$. But $f_d$ must be non-negative, so $d$ must be odd. Then $\phi^{d-1}>0$. Since $f_{d-1} = \phi^{d-1}\tp{x+y\tp{d-1}}\ge 0$, it holds that $x+y\tp{d-1}\ge 0$. Since $d > 1$, we can similarly deduce that $x+y\tp{d-2}\le 0$. This contradicts that $x>0$ and $x+y\tp{d-1}\ge 0$.
  \end{itemize}

  We have showed that there is a matrix $M'\in \=C^{2\times 2}$ such that the zeros of $P_{f\cdot M'^{\otimes d}}\tp{z}$ are in the open left half-plane.
  Since a polynomial has only a finite number of zeros, there is a constant $\varepsilon>0$ that $P_{f\cdot M'^{\otimes d}}\tp{z}$ is $H_{\varepsilon}$-stable.
  It holds that $M'M'^T=\trans{1+w^2}{0}{0}{1+w^2}=\tp{1+w^2}I_2$ where $1+w^2>0$ as $w\in\=R$.
  Let $M=\frac{1}{\sqrt{1+w^2}}M'$, and clearly $M\in\*O_2(\=C)$.
  Since $P_{f\cdot M'^{\otimes d}}\tp{z}=\tp{1+w^2}^{d/2} P_{f\cdot M^{\otimes d}}\tp{z}$,
  $P_{f\cdot M^{\otimes d}}\tp{z}$ has the same set of zeros as $P_{f\cdot M'^{\otimes d}}\tp{z}$.
  So $P_{f\cdot M^{\otimes d}}$ is also $H_\varepsilon$-stable for some $\varepsilon>0$.
\end{proof}

\subsection{\texorpdfstring{$b^2-4ac<0$}{b2-4ac<0}}

When the characteristic polynomial of $f$ has two distinct complex roots, we show that the local
polynomial of $f$ itself is $H_\eps$-stable.

\begin{lemma} \label{lem:delta-1} Let $f=\sqtp{f_0,f_1,\ldots,f_d}\in \+F^*_{a,b,c}$ with
  $b^2-4ac<0$, then $P_f\tp{z}$ is $H_{\varepsilon}$-stable for some $\varepsilon>0$.
\end{lemma}

\begin{proof}
  It holds that $c\neq 0$ since otherwise $b^2-4ac\ge 0$.  Since $c\neq 0$ and $b^2-4ac<0$, it
  follows from \Cref{prop:2order} that $f_k=x\phi^k + y\ol{\phi}^k$ for $0\le k\le d$, where
  $\phi,\ol\phi$ are the two conjugate roots of the polynomial $cz^2+bz+a=0$ and $x,y\in \=R$ are
  constants.  Clearly $x+y=f_0$ and $x\phi + y\ol\phi=f_1$.  Since $f_0$ is real, it holds that
  $\Im\tp{y}= -\Im\tp{x}$.  Since $f_1$ is real and
  $f_1=x\phi+y\ol{\phi}=\tp{x+y}\Re\tp{\phi} + i\tp{x-y}\Im\tp{\phi}$, it holds that
  $\Re\tp x = \Re\tp y$. Thus $y=\ol{x}$ and $f_k=x\phi^k+\ol{x}\ol{\phi}^k$ for $0\le k\le d$. We
  write $f=x\tp{1,\phi}^{\otimes d}+\ol x\tp{1,\ol\phi}^{\otimes d}$ and
  \begin{align*}
    P_f(z)=x\tp{1+\phi z}^d + \ol{x}\tp{1+\ol\phi  z}^d.
  \end{align*}
  The zeros of $P_f\tp{z}$ must satisfy
  \begin{align}
    \abs{x}\cdot \abs{1+\phi z}^d = \abs{\ol x}\cdot \abs{1+\ol \phi z}^d. \label{eq:zeros4}
  \end{align}
  Note that $\phi\neq 0$, and $x\neq 0$ since otherwise $\ol{x}=0$ and $f$ would be
  $\sqtp{0,0,\ldots,0}$. So the equation \eqref{eq:zeros4} is equivalent to
  \begin{align*}
    \abs{z-\tp{-\frac{1}{\phi}}}=\abs{z-\tp{-\frac{1}{\ol \phi}}}.
  \end{align*}
  Since $-\frac{1}{\phi}$ and $-\frac{1}{\ol \phi}$ are the complex conjugates of each other, the
  roots of this equation and thus the zeros of $P_f\tp{z}$ must lie on the real axis.
  On the other hand, if $z\ge 0$
  \begin{align*}
    P_f\tp{z} = \sum_{k=0}^d \binom{n}{k} f_k \cdot z^k >0, 
  \end{align*}
  since $f_0>0$.
  Thus the zeros of $P_f\tp{z}$ are negative reals.
  Since a polynomial has only a finite number of zeros,
  there is a constant $\varepsilon>0$ such that $P_f\tp{x}$ is $H_{\varepsilon}$-stable.
\end{proof}

\section{Exceptional cases}\label{sec:exception}

\Cref{sec:2nd-rec} covered all signatures in $\+F_{a,b,c}$ unless $f_0=f_d=0$.
We discuss the remaining cases in this section.
We will classify all of them, but the approximation complexity in one case is still open.

Let $b\in\=R$, and define $\+A_b$ to be the following class
\begin{align*}
  \set{ [f_0,f_1,\dots,f_{d}] \mid \forall 0\le k\le d-2,\;\frac{b^2}{4\cos^2\frac{\pi}{d}} f_k + b f_{k+1} + f_{k+2} = 0,\;f_0=0\text{ and }f_1>0}.
\end{align*}
Notice that $\+A_b$ is a special case of $\+F_{a,b,c}$ except that the parameter $a$ depends on the arity $d$.
In fact, if $f\in\+A_b$, then we can scale $f$ so that $f$ has the following form
\begin{align*}
  \left[0,\lambda\sin\frac{\pi}{d},\lambda^2\sin\frac{2\pi}{d},\dots,\lambda^{d-1}\sin\frac{(d-1)\pi}{d},0\right],
\end{align*}
for $\lambda=-\frac{b}{2\cos\frac{\pi}{d}}>0$. (Recall that $b< 0$.)
Namely, $f_i=\lambda^{i}\sin\frac{i\pi}{d}$.

\begin{lemma}\label{lem:exceptions}
  Let $f=[f_0,f_1,\dots,f_{d}] \in \+F_{a,b,c}$ for some $d\ge 3$.
  If $f_0=f_d=0$, then there are three possibilities:
  \begin{itemize}
  \item[I.] $f\in\+A_b$ for some $b < 0$;
  \item[II.] $[f_0,f_1,\dots,f_{d}]$ is of form $[0,*,0,0,\dots,0]$ or its reversal
    $[0,0,\dots,0,*,0]$;
  \item[III.] $[f_0,f_1,\dots,f_{d}]$ is of form
    $\lambda [0,1,0,\mu, 0, \mu^2, \dots,0,\mu^{\frac{d-2}{2}},0]$ for
    some $\lambda, \mu> 0$ and even $d$.
  \end{itemize}
\end{lemma}

\begin{proof}
  We start by considering the case $c=0$. Then $af_k+bf_{k+1}=0$ for every $0\le
  k\le d-2$. It is easy to verify that $f$ is identically $0$ as $f_0=f_d=0$, which belongs to type II.
  Thus, we may assume that $c\ne 0$ and normalise $c$ to $1$ in the following.
  There are two further cases depending on whether $b^2-4a=0$.

  The first case is when $b^2-4a\ne 0$. It follows from Proposition~\ref{prop:2order} that
  $f_0=x+y=0$ and $f_d=x\phi_1^d+y\phi_2^d=0$. These two identities together imply
  \[
    x\tp{\phi_1^d-\phi_2^d}=0,
  \]
  which further implies either $x=y=0$ (and therefore $f_k=0$ for all $k$) or
  $\phi_1^d=\phi_2^d$. We only need to discuss the case when $\phi_1^d=\phi_2^d$ and $x\ne 0$.
  There are two possibilities.
  \begin{itemize}
    \item [(1)] If $\frac{\phi_1}{\phi_2}\in\=R$, then $\phi_1=-\phi_2$ as $b^2\neq 4a$.
      It implies that $d$ is even. This is type III.
    \item [(2)] Otherwise, $\frac{\phi_1}{\phi_2}\not\in\=R$.
      In this case, $b^2-4a<0$ and $\phi_1$ and $\phi_2$ are conjugate of each other.
      By swapping $\phi_1$ and $\phi_2$ if necessary, we may assume that $0<\arg\phi_1<\pi$.
      Then there exists some integer $0< t< d$, $t\neq d/2$, so that $\arg\phi_1 = \frac{t\pi}{d}$ and $\frac{\phi_1}{\phi_2}=e^{\frac{2t\pi}{d}i}\not\in\=R$.
      Since $a>b^2/4\ge 0$, $\abs{\phi_1}=\abs{\phi_2}=\sqrt{a}$, and
     \begin{align*}
        f_k=x\tp{\phi_1^k-\phi_2^k} = 2 x \cdot a^{\frac{k}{2}}\tp{\sin\frac{tk\pi}{d}} i.
      \end{align*}
      Recall that we have the further requirement $f_k\ge 0$ for every $0\le k\le d$.
      For $k=1$, as $0<t<d$, $\sin\frac{t\pi}{d}>0$, and thus $x$ must lie on the negative imaginary axis.
      Then, it must be that $\sin\frac{tk\pi}{d}\ge 0$ for all $0\le k\le d$.
      If $t>1$, then taking $k=\floor{\frac{d}{t}}+1\le d$ implies a contradiction.
      Thus $t=1$.

      The assumption $0<\arg\phi_1<\pi$ implies that $\cos\frac{\pi}{d}=\frac{-b}{2\sqrt{a}}>0$.
      Thus, $b< 0$ and $a=\frac{b^2}{4\cos^2\frac{\pi}{d}}$.
      This verifies that $f$ is of type I.
  \end{itemize}

  At last we turn to the case that $b^2-4a=0$.
  It follows from Proposition~\ref{prop:2order} that $f_k=x\phi^k+yk\phi^{k-1}$ where $\phi=-b/2$.
  Then $f_0=0$ means that $x=0$, and $f_d=0$ means that $y\phi^{d-1}=0$.
  Thus either $y=0$ or $\phi=0$, and any of the two cases implies that
  $f$ is of type II.
\end{proof}

\newcommand{\EO}{\textnormal{\textsc{ExactOne}}}
\newcommand{\PM}{\textnormal{\#PM}}
\newcommand{\AP}{\le_{\textnormal{\texttt{AP}}}}
\newcommand{\APge}{\ge_{\textnormal{\texttt{AP}}}}
\newcommand{\leGadget}{\le_{\textnormal{\texttt G}}}
\newcommand{\geGadget}{\ge_{\textnormal{\texttt G}}}

Next we show that type II and type III signatures are equivalent to approximately counting perfect matchings in general graphs.
Denote by $\EO_d$ the function $[0,1,0,\dots,0]$ of arity $d$,
and by $\+EO$ the (infinite) set $\set{\EO_d\mid d\in\=N^+}$.
Then $\Holant(\+EO)$ is the problem of counting perfect matchings in a graph, denoted $\PM$.
(There is one function per each degree/arity. So the mapping from vertices to functions is obvious for the infinite set $\+EO$.)

For type III signatures, since multiplying by a constant does not change the complexity, we may assume that $\lambda=\sqrt{\mu}$.
Then $f=[0,\lambda,0,\lambda^3,0,\dots,\lambda^{d-1},0]$ with $\lambda>0$.
We will assume $\lambda<1$.
This is because that if $\lambda=1$, then the problem is tractable exactly, (see, for example, \cite{CGW16})
and if $\lambda>1$, then taking its reversal makes $\lambda<1$.
We adopt the approximation-preserving reduction $\AP$ from \cite{DGGJ04},
and use $\leGadget$ to denote gadget reductions, which is a special form of $\AP$.

\begin{lemma}  \label{lem:PM4<=TypeIII}
  Let $d\ge 4$ be an even integer, and $0<\lambda<1$.
  If $f=[0,\lambda,0,\lambda^3,0,\dots,\lambda^{d-1},0]$ of arity $d$, then
  \[\Holant(\EO_4)\AP\Holant(f).\]
\end{lemma}
\begin{proof}
  Applying a holographic transformation by $T=\trans{1}{0}{0}{\lambda}$,
  we have that
  \begin{align*}
    \Holant(f) & \equiv \holant{f\cdot \tp{T^{-1}}^{\otimes d}}{T^{\otimes 2}\cdot=_2} \\
    & \equiv \holant{[0,1,0,1,0,\dots,1,0]}{[1,0,\mu]},
  \end{align*}
  where $0<\mu=\lambda^2<1$.
  Thus $\Holant(f)$ is to count the number of odd subgraphs with edge weight $\mu$ in a $d$-regular graph.
  Notice that doing a self-loop simply reduces the degree of a vertex by $2$,
  while leaving the constraint on the vertex still requires ``odd-degrees''.
  Thus, with enough self-loops, we may simulate a binary disequality $[0,1,0]$
  as well as an arity-$4$ signature $[0,1,0,1,0]$ on the left hand side of the bipartite Holant formulation.

  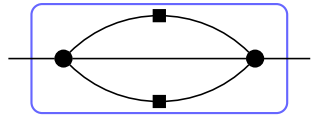
\begin{figure}[h]
    \centering
     \begin{tikzpicture}[scale=\scale,transform shape,node distance=\nodeDist,semithick]
      \node[external] (0)              {};
      \node[internal] (1) [right of=0] {};
      \node[external] (2) [right of=1] {};
      \node[external] (3) [right of=2] {};
      \node[internal] (4) [right of=3] {};
      \node[external] (5) [right of=4] {};
      \path (0) edge                          node[near end]   (e1) {}               (1)
            (1) edge[out= 45, in= 135]        node[square]     (e2) {}               (4)
                edge[out=  0, in= 180]                                               (4)
                edge[out=-45, in=-135]        node[square]     (e3) {}               (4)
            (4) edge                          node[near start] (e4) {}               (5);
      \begin{pgfonlayer}{background}
       \node[inner sep=2pt,transform shape=false,draw=\borderColor,thick,rounded corners,fit = (e1) (e2) (e3) (e4)] {};
      \end{pgfonlayer}
    \end{tikzpicture}
    \caption{A gadget for type III signatures, where squares are $[0,1,0]$, and circles are $[0,1,0,1,0]$. All edges are $[1,0,\mu]$.}
    \label{fig:gadget}
  \end{figure}

  Consider the gadget in \Cref{fig:gadget}.
  Then, it is easy to verify that the effective binary function is $(2\mu^2+2\mu^3)[1,0,1]$ on the left hand side.
  Finally, with $[1,0,1]$ on the left, we can form a path of length $n$, and the resulting binary function is $[1,0,\mu^n]$ on the right.
  More formally, we have the following chain of reductions:
  \begin{align*}
    \holant{[0,1,0,1,0,\dots,1,0]}{[1,0,\mu]} & \geGadget \holant{[0,1,0],[0,1,0,1,0]}{[1,0,\mu]} \\
    & \geGadget \holant{[1,0,1],[0,1,0,1,0]}{[1,0,\mu]} \\
    & \geGadget \holant{[0,1,0,1,0]}{[1,0,\mu^n]}.
  \end{align*}
  The last problem is counting odd subgraphs with $\mu^n$ edge weights in $4$-regular graphs and $\mu<1$.
  Now, one moment's reflection realises that odd subgraphs with exponentially small edge weights is approximately perfect matchings,
  which finishes the reduction.
\end{proof}

Similar ideas can also handle the last case in \Cref{lem:delta1}, after taking its reversal and renaming $\lambda$.

\begin{lemma}  \label{lem:PM3<=TypeIII}
  Let $d\ge 3$ be an odd integer, and $0<\lambda<1$.
  If $f=[0,\lambda,0,\lambda^3,0,\dots,\lambda^{d}]$ of arity $d$, then
  \[\Holant(\EO_3)\AP\Holant(f).\]
\end{lemma}
\begin{proof}
  As in the proof of \Cref{lem:PM4<=TypeIII},
  we do the same holographic transformation by $T=\trans{1}{0}{0}{\lambda}$:
  \begin{align*}
    \Holant(f) & \equiv \holant{[0,1,0,1,0,\dots,1]}{[1,0,\mu]},
  \end{align*}
  where $0<\mu=\lambda^2<1$.
  Once again, with sufficiently many self-loops,
  we get $[0,1,0,1]$ and $[0,1]$ on the left hand side.
  Connecting $[0,1]$ back to $[0,1,0,1]$ through $[1,0,\mu]$ yields $\mu[1,0,1]$ on the left.
  Thus, similar to the proof of \Cref{lem:PM4<=TypeIII}, we can simulate $[1,0,\mu^n]$ on the right.
  More formally, we have the following chain of reductions:
  \begin{align*}
    \holant{[0,1,0,1,0,\dots,1]}{[1,0,\mu]} & \geGadget \holant{[0,1],[0,1,0,1]}{[1,0,\mu]} \\
    & \geGadget \holant{[1,0,1],[0,1,0,1]}{[1,0,\mu]} \\
    & \geGadget \holant{[0,1,0,1]}{[1,0,\mu^n]} \\
    & \APge \Holant(\EO_3). \qedhere
  \end{align*}
\end{proof}

On the other hand, we have the following lemma.

\begin{lemma}  \label{lem:typeIII<=PM}
  Let $d\ge 3$ be an integer and $0<\lambda<1$.
  Let $f=[0,\lambda,0,\lambda^3,0,\dots]$ be a symmetric signature of arity $d$.
  Then \[\Holant(f)\AP\PM.\]
\end{lemma}
\begin{proof}
  First, by the same holographic transformations as in the proofs of \Cref{lem:PM4<=TypeIII} and \Cref{lem:PM3<=TypeIII},
  \[\Holant(f)\equiv\holant{[0,1,0,1,0,\dots]}{[1,0,\mu]},\]
  where $\mu=\lambda^2>0$.

  Consider the gadget in \Cref{fig:gadget-eq}, where all vertices are the ``exact one'' function, namely $[0,1,0,\dots,0]$.
  It is easy to see that this gadget is equivalent to a weighted equality $[1,0,\frac{n_2}{n_1}]$.
  Thus we can use it to arbitrarily closely approximate $[1,0,\mu]$ by tuning $n_1$ and $n_2$ for any $\mu>0$.

  \begin{figure}[h]
  \begin{minipage}{.55\textwidth}
    \centering
     \begin{tikzpicture}[scale=\scale,transform shape,node distance=\nodeDist,semithick]
      \node[external] (0)              {};
      \node[internal] (1) [right of=0, label=90:{\Large $u$}] {};
      \node[internal] (2) [right of=1, label=90:{\Large $u'$}] {};
      \node[internal] (3) [right of=2, label=90:{\Large $v'$}] {};
      \node[internal] (4) [right of=3, label=90:{\Large $v$}] {};
      \node[external] (5) [right of=4] {};
      \path (0) edge                          node[near end]   (e1) {}               (1)
            (1) edge[out= 45, in= 135]        (2)
                edge[out= 20, in= 180, opacity=0]        node[opacity=1]{\Large \vdots}(2)
                edge[out=-45, in=-135]        (2)
            (2) edge[out= 60, in= 120]        node (e2) {} (3)
                edge[out= 30, in= 150]                     (3)
                edge[out= -30, in= 150, opacity=0]        node [opacity=1]{\Large \vdots} (3)
                edge[out=-60, in=-120]        node (e3) {} (3)
            (3) edge (4)
            (4) edge                          node[near start] (e4) {}               (5);
      \begin{pgfonlayer}{background}
        \node[inner sep=2pt,transform shape=false,draw=\borderColor,thick,rounded corners,fit = (e1) (e2) (e3) (e4)] {};
      \end{pgfonlayer}
    \end{tikzpicture}
    \caption{A gadget for weighted equalities. There are $n_1$ edges between $u$ and $u'$, and $n_2$ edges between $u'$ and $v'$.}
    \label{fig:gadget-eq}
  \end{minipage}%
  \begin{minipage}{.45\textwidth}
    \centering
     \begin{tikzpicture}[scale=\scale,transform shape,node distance=\nodeDist,semithick]
        \draw (90:1cm) node [internal] (v1) {};
        \draw (210:1cm) node [internal] (v2) {};
        \draw (330:1cm) node [internal] (v3) {};
        \draw (90:2.15cm) node [external] (u1) {};
        \draw (210:2.15cm) node [external] (u2) {};
        \draw (330:2.15cm) node [external] (u3) {};
        \draw (v1) -- (v2);
        \draw (v2) -- (v3);
        \draw (v3) -- (v1);
        \foreach \x in {1,2,3}
        {
          \draw (v\x) -- (u\x);
        }

      \begin{pgfonlayer}{background}
       \node[inner sep=4pt,transform shape=false,draw=\borderColor,thick,rounded corners,fit = (v1) (v2) (v3)] {};
      \end{pgfonlayer}
    \end{tikzpicture}
    \caption{A gadget to create $[0,1,0,1]$.}
    \label{fig:gadget-odd}
  \end{minipage}
  \end{figure}

  In addition, consider the gadget in \Cref{fig:gadget-odd}, where, once again, all vertices are $[0,1,0,0]$.
  The resulting signature is $[0,1,0,1]$.

  \begin{figure}[h]
     \begin{tikzpicture}[scale=\scale,transform shape,node distance=\nodeDist,semithick]
        \draw (0,0) node [internal] (v1) {};
        \draw (1.5,0) node [internal] (v2) {};
        \draw (3,0) node [external] (v3) {\Huge \dots};
        \draw (4.5,0) node [internal] (v4) {};
        \draw (6,0) node [internal] (v5) {};
        \draw (-1,1) node [external] (e11) {};
        \draw (-1,-1) node [external] (e12) {};
        \draw (1.5,1) node [external] (e2) {};
        \draw (4.5,1) node [external] (e4) {};
        \draw (7,1) node [external] (e51) {};
        \draw (7,-1) node [external] (e52) {};

        \draw (v1) -- (v2);
        \draw (v2) -- (v3);
        \draw (v3) -- (v4);
        \draw (v4) -- (v5);

        \draw (v1) -- (e11);
        \draw (v1) -- (e12);
        \draw (v2) -- (e2);
        \draw (v4) -- (e4);
        \draw (v5) -- (e51);
        \draw (v5) -- (e52);
      \begin{pgfonlayer}{background}
       \node[inner sep=4pt,transform shape=false,draw=\borderColor,thick,rounded corners,fit = (v1) (v5)] {};
      \end{pgfonlayer}
    \end{tikzpicture}
    \caption{A gadget to create $[1,0,1,0,\dots,1]$ or $[0,1,0,1,\dots,0]$. }
    \label{fig:gadget-odd2}
  \end{figure}
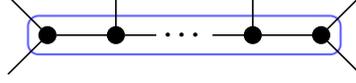

  A simple calculation verifies that a sequence of $d$ signatures $[0,1,0,1]$ connected together, as in \Cref{fig:gadget-odd2},
  yields a signature $[0,1,0,\dots,1,0]$ of arity $d+2$ if $d$ is odd,
  or a signature $[1,0,1,0,\dots,1]$ of arity $d+2$ if $d$ is even.
  In the even case, to get $[0,1,0,1,\dots,0]$, we simply connect one of its dangling edges with $[0,1,0]$.
  Formally, we have the following sequence of reductions:
  \begin{align*}
    \holant{[0,1,0,1,0,\dots]}{[1,0,\mu]} & \leGadget \Holant(\{[0,1,0,1,0,\dots],[1,0,\mu]\}) \\
    & \le_{\texttt{AP}} \PM. \qedhere
  \end{align*}
\end{proof}

\Cref{lem:PM4<=TypeIII}, \Cref{lem:PM3<=TypeIII}, and \Cref{lem:typeIII<=PM} together imply the following:
\begin{align}  \label{eqn:PM4-TypeIII-PM}
  \Holant(\EO_4)&\AP\Holant(f)\AP\PM, \text{ if $d$ is even,}\\
  \Holant(\EO_3)&\AP\Holant(f)\AP\PM, \text{ if $d$ is odd,} \label{eqn:PM3-TypeIII-PM}
\end{align}
where $f=[0,\lambda,0,\lambda^3,0,\dots]$ for some $0<\lambda<1$ has arity $d\ge 3$.
Note that $\Holant(\EO_3)$ or $\Holant(\EO_4)$ is just an alias of counting perfect matchings in $3$- or $4$-regular graphs,
which is equivalent to $\PM$ in approximation.
This is a folklore fact, and is shown in the next couple of lemmas.

\begin{lemma}  \label{lem:EO3<=EO4}
  $\Holant(\EO_3)\AP\Holant(\EO_4)$.
\end{lemma}
\begin{proof}
  Note that a self-loop on $[0,1,0,0,0]$ gives $[0,1,0]$,
  and connecting it back to $[0,1,0,0,0]$ yields $[1,0,0]$.
  Thus,
  \begin{align*}
    \Holant([0,1,0,0,0],[1,0,0])\leGadget \Holant([0,1,0,0,0]).
  \end{align*}
  Given an instance $G$ (namely a $3$-regular graph) of $\Holant([0,1,0,0])$,
  consider a disjoint union of $G$ and its copy $G'$.
  We add a new vertex $u$ for each pair $v$ and $v'$, and connect $u$ to both $v$ and $v'$.
  Now all original vertices in $G$ and $G'$ have degrees exactly $4$.
  Put $[0,1,0,0,0]$ on all these vertices, and $[1,0,0]$ on all new vertices.
  It is easy to see that the partition function of this new instance is the square of the number of perfect matchings of $G$.
  Thus, we have the following reduction chain:
  \begin{align*}
    \Holant([0,1,0,0]) & \AP \Holant([0,1,0,0,0],[1,0,0]) \\
    & \leGadget \Holant([0,1,0,0,0]). \qedhere
  \end{align*}
\end{proof}

However, approximate counting perfect matchings in $3$-regular graphs is as hard as that in general graphs.

\begin{lemma}  \label{lem:PM<=EO3}
  $\PM\AP\Holant(\EO_3)$.
\end{lemma}
\begin{proof}
  Consider the gadget in \Cref{fig:gadget-EOd}.

  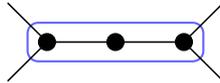
\begin{figure}[h]
     \begin{tikzpicture}[scale=\scale,transform shape,node distance=\nodeDist,semithick]
        \draw (0,0) node [internal] (v1) {};
        \draw (1.5,0) node [internal] (v2) {};
        \draw (3,0) node [internal] (v5) {};
        \draw (-1,1) node [external] (e11) {};
        \draw (-1,-1) node [external] (e12) {};
        \draw (4,1) node [external] (e51) {};
        \draw (4,-1) node [external] (e52) {};

        \draw (v1) -- (v2);
        \draw (v2) -- (v5);

        \draw (v1) -- (e11);
        \draw (v1) -- (e12);
        \draw (v5) -- (e51);
        \draw (v5) -- (e52);
      \begin{pgfonlayer}{background}
       \node[inner sep=4pt,transform shape=false,draw=\borderColor,thick,rounded corners,fit = (v1) (v5)] {};
      \end{pgfonlayer}
    \end{tikzpicture}
    \caption{A gadget to create $\EO_d$.}
    \label{fig:gadget-EOd}
  \end{figure}

  Notice that if we put $[0,1,0,0]$ on the two degree three vertices, and $[0,1,0]$ on the middle vertex,
  the resulting signature is $[0,1,0,0,0]$.
  More generally, if we replace one of the degree three vertex by $\EO_d$,
  then the resulting signature is $\EO_{d+1}$.
  Namely, using this gadget, we can simulate the whole set of $\+EO$, and
  \begin{align*}
    \PM & \leGadget \Holant([0,1,0,0],[0,1,0]).
  \end{align*}

  Moreover, a self-loop on $[0,1,0,0]$ gives $[0,1]$,
  and connecting back to it gives $[1,0,0]$.
  By using the same squaring trick in \Cref{lem:EO3<=EO4},
  we can use $[1,0,0]$ as $[1,0]$.
  Thus, we have the following reduction chain:
  \begin{align*}
    \PM & \leGadget \Holant([0,1,0,0],[0,1,0]) \\
    & \leGadget \Holant([0,1,0,0],[1,0]) \\
    & \AP \Holant([0,1,0,0],[1,0,0]) \\
    & \leGadget \Holant([0,1,0,0]). \qedhere
  \end{align*}
\end{proof}

Holant problems defined by type II signatures are counting perfect matchings in $d$-regular graphs.
Clearly, by doing sufficiently many self-loops, either $\Holant(\EO_3)\AP\Holant(\EO_d)$ or $\Holant(\EO_4)\AP\Holant(\EO_d)$,
depending on the parity of $d$.
Thus, combining this fact with \Cref{lem:EO3<=EO4}, \Cref{lem:PM<=EO3}, \eqref{eqn:PM4-TypeIII-PM} and \eqref{eqn:PM3-TypeIII-PM},
we have the following result.

\begin{lemma}  \label{lem:TypeII-TypeIII}
  Let $f=[0,1,0,\lambda^2,0,\dots]$ for some $0\le\lambda<1$.
  Then \[\Holant(f)\equiv_{\textnormal{\texttt{AP}}}\PM.\]
\end{lemma}

Notice that in \Cref{lem:TypeII-TypeIII} we manipulate the form a little bit so that it cover type II and type III in \Cref{lem:exceptions},
as well as the last case in \Cref{lem:delta1}.

\section{Proof of main theorems}

We are now ready to assemble all the ingredients to prove our main theorems. We restate
\Cref{thm:main} for convenience.

{\renewcommand{\thetheorem}{\ref{thm:main}}
  \begin{theorem}
    Let $f=[f_0,f_1,\dots,f_{d}]$ be a symmetric constraint function of arity $d\ge 3$ satisfying generalised second-order recurrences,
    and $f_i\ge 0$ for all $0\le i\le d$.
    There is a fully polynomial-time (deterministic or randomised) approximation algorithm for $\Holant(f)$,
    unless, up to a non-zero factor, $f$ or its reversal is in one of the following form:
    \begin{itemize}
      \item $[0,\lambda\sin\frac{\pi}{d},\lambda^{2}\sin\frac{2\pi}{d},\dots,\lambda^{i}\sin\frac{i\pi}{d},\dots,0]$ for some $\lambda>0$;
      \item $[0,1,0,\lambda,0,\dots,0,\lambda^{\frac{d-2}{2}},0]$ if $d$ is even, or $[0,1,0,\lambda,0,\dots,0,\lambda^{\frac{d-1}{2}}]$ if $d$ is odd, for some $0\le\lambda < 1$.
    \end{itemize}
    Moreover, in the latter case, approximating $\Holant(f)$ is equivalent to approximately counting perfect matchings in general graphs.
  \end{theorem} \addtocounter{theorem}{-1} }
\begin{proof}
  We apply \cref{lem:delta1}, \cref{lem:delta0} and \cref{lem:delta-1}. Then one of followings must
  happen
  \begin{enumerate}
    \item \label{case:f0=fd=0} $f_0=f_d=0$; or
    \item \label{case:odd} $f$ or $\ol f$ is $[1,0,\lambda^2,0,\lambda^4,0,\dots]$ for some $\lambda > 1$ and has an odd arity; or
    \item \label{case:tractable} $\Holant\tp{f}$ can be solved exactly in polynomial-time; or
    \item \label{case:ferro} there is an invertible matrix $M\in\=C^{2\times 2}$ such that
      $\holant{fM^{\otimes d}}{\tp{M^{-1}}^{\otimes 2}\cdot\tp{=_2}}$ is a ferromagnetic two-spin
      system; or
    \item \label{case:zeros} there is an orthogonal matrix $M\in \*O_2(\=C)$ such that either
      $P_{f\cdot M^{\otimes d}}\tp{z}$ or $P_{\ol f\cdot M^{\otimes d}}\tp{z}$ is
      $H_{\varepsilon}$-stable for some $\varepsilon>0$, where $\ol{f}$ is the reversal of $f$.
  \end{enumerate}

  We are done in \Cref{case:tractable}, as well as in \Cref{case:zeros} by \Cref{prop:orthogonal} and \Cref{thm:roots-Holant}.
  In \Cref{case:ferro}, we invoke the FPRAS by Jerrum and Sinclair~\cite{JS93}.
  In \Cref{case:f0=fd=0} and \Cref{case:odd},
  we are in the desired form of the theorem by \Cref{lem:exceptions}.
  (In case $\mu>1$ in \Cref{lem:exceptions}, we can take its reversal so that $\mu<1$,
  and if $\mu=1$, then exact counting is tractable \cite{CGW16}.)
  Finally, the approximation complexity of $[0,1,0,\lambda,0,\lambda^2,0,\dots]$ signatures is handled in \Cref{lem:TypeII-TypeIII}.
\end{proof}

\begin{remark}
  It is worth noticing that our algorithm applies beyond regular graphs. In fact, for any finite
  family of signatures $\+F$, we can define $\Holant\tp{\+F}$ as the problem of computing the
  partition function on a graph $G$, where each vertex $v$ of $G$ is associated with a function
  $f_v\in\+F$. It is straightforward to adapt the algorithm described in the proof of
  \cref{thm:main} for one to solve $\Holant\tp{\+F}$\footnote{The main adaptation is to show that
    $Z_i(G)$ is still a BIGCP when more than one constraint function are present. Since $\+F$ is
    finite, we can therefore view functions in $\+F$ as colors and enumerate \emph{vertex colored
      induced subgraphs} instead of ordinary induced subgraphs in the proof of
    \cref{lem:Holant-BIGCP}. Similar technique already appreared in~\cite{PR17a}}. It is not hard to
  see the adaptation provides an efficient approximation algorithm for $\Holant\tp{\+F}$ as long as
  there exists an orthogonal matrix $M\in\*O_2(\=C)$ and $\eps>0$ such that $P_{f\cdot M^{\otimes d}}$ is
  $H_\eps$-stable for every $f\in\+F$, where $d$ is the arity of $f$. For example, we can let $\+F$ be the family of
  signatures for matchings up to arity $d$, or the family of signatures for edge covers up to arity $d$.
  Therefore, our algorithm recovers a number of previously known deterministic approximation algorithms for special
  cases of Holant problems, such as counting matchings \cite{BGKNT07,PR17a} and counting edge covers in bounded degree graphs \cite{LLL14}.

  On the other hand, even for the same tuple $(a,b,c)$,
  signatures in $\+F_{a,b,c}$ may require different $M$ to be $H_{\eps}$-stable.
  It is not clear how to obtain an algorithm in such cases.
\end{remark}

We deduce \cref{thm:main-cubic} from \cref{thm:main} by noting that all ternary signatures
satisfy generalised second-order recurrence relations. Therefore, we only need to deal with the case
where $f=[0,a,b,0]$ for some $a,b > 0$.

We design an FPRAS for $\Holant\tp{f}$ using the machinery called ``winding'' developed
in~\cite{McQ13,HLZ16}. We sketch the construction here without getting into too much technical
details, which is out of the scope of the current paper.  We break every edge into two half edges,
and then simulate a Markov chain whose state space consists of all consistent edge assignments and
assignments with at most two inconsistencies.  It has been shown by McQuillan~\cite{McQ13} that the
Markov chain mixes rapidly as long as the signature $f$ is \emph{windable}.  It is then straightforward to
use the algebraic characterization of windable functions in~\cite{HLZ16} to verify that every
function of the form $[0,a,b,0]$ with non-negative $a,b$ is windable.  At last, it is trivial to
check that, using the notations in \cite{McQ13}, the signature $[0,a,b,0]$ is \emph{strictly
  terraced} when both $a,b>0$.  This fact implies that the ratio between the total weight of nearly
consistent assignments and that of consistent assignments can be bounded by a polynomial in the size
of the instance.  Therefore, we obtain an efficient Gibbs sampler for $\Holant\tp{f}$, which can be
turned into an FPRAS to compute the partition function via self-reduction~\cite{JVV86}.

The remaining open case in \Cref{thm:main} is when $f\in\+A_b$ of arity $d\ge 3$.
Numerical evidences suggest that these signatures are windable, via the criteria in~\cite{HLZ16}.
We conjecture that this is indeed the case, which would imply FPRAS for computing the partition functions of type I signatures,
since this class is ``strictly terraced'' in the language of \cite{McQ13}.

\section*{Acknowledgements}

Part of the work was done while HG, CL, and CZ were visiting the Institute of Theoretical Computer Science, Shanghai University of Finance and Economics,
and we would like to thank their hospitality.
We thank Leslie Ann Goldberg and Mark Jerrum for pointing out a bug in \Cref{lem:delta1} in an earlier version.

\bibliographystyle{alpha} \bibliography{roots}

\end{document}